\documentclass[runningheads, envcountsame, a4paper]{article}

\usepackage{amssymb}
\usepackage{amsmath}
\usepackage{ amsfonts} 
\usepackage{ theorem}  
\usepackage{times}
\usepackage{stmaryrd}
\usepackage{mathrsfs}
\usepackage{latexsym}
\usepackage{hhline}

\usepackage{color}
\usepackage{graphicx}
\usepackage[skip=0pt]{caption}
\usepackage{subcaption}

\usepackage{tikz}
\usepackage{lineno}
\usepackage{todonotes}
\usepackage{bm}

\theoremstyle{plain}
\newtheorem{theorem}{Theorem}

\newtheorem{lemma}[theorem]{Lemma}

\newtheorem{proposition}[theorem]{Proposition}

\newtheorem{definition}[theorem]{Definition}
\newtheorem{example}[theorem]{Example}
\newtheorem{remark}[theorem]{Remark}

\newenvironment{proof}{\noindent{\bf Proof:\/}}{$\Box$\vskip 0.1in}

 \def\B{{\tilde{\delta}}}
 
 \def\tdist{{\rm dist}_{\sim}}

\begin{document}

\markboth{Anselmo, Castiglione, Flores, Giammarresi, Madonia, Mantaci}
{  Isometric Words based on Swap and Mismatch Distance}

\title{\Large Characterization of Isometric Words
{based on} 
\\Swap and {Mismatch} Distance		
}
\author{Author One$^1$\thanks{Author One was partially supported by Grant XXX} \and Author Two$^2$ \and Author Three$^1$}

\date{
	$^1$Organization 1 \\ \texttt{\{auth1, auth3\}@org1.edu}\\%
	$^2$Organization 2 \\ \texttt{auth3@inst2.edu}\\[2ex]%
}

\author{ M. Anselmo$^1$ \and  G. Castiglione$^2$ \and  M. Flores$^{2}$ \and   D. Giammarresi$^3$ \and   M. Madonia$^4$ \and S. Mantaci$^2$
  }%

  \date{	
$^1${ \normalsize Dipartimento di Informatica,
Universit\`a di Salerno, 
Italy. 
\\
{\tt
\{manselmo\}@unisa.it}} \\
$^2${ \normalsize Dipartimento di Matematica e Informatica, Universit\`a di Palermo, 
Italy  
{\tt \{giuseppa.castiglione, sabrina.mantaci, manuela.flores\}@unipa.it }} \\
$^3${ \normalsize Dipartimento di Matematica.
 Universit\`a  Roma ``Tor Vergata''
  Italy. 
{\tt giammarr@mat.uniroma2.it}}\\
$^4${\normalsize Dipartimento di Matematica eInformatica, Universit\`a  di
  Catania, 
  Italy.
  {\tt madonia@dmi.unict.it}} 
}

 \maketitle

\begin{abstract}
In this paper we consider an edit distance with swap and mismatch operations, called  {\em tilde-distance}, and introduce the corresponding definition of {\em tilde-isometric} word. 
Isometric words are classically defined with respect to Hamming distance and combine the notion of edit distance  with the property that a word does not appear as factor in other words.
A word $f$ is said tilde-isometric if, for any pair of $f$-free words $u$ and $v$, there exists a transformation from $u$ to $v$ via the related edit operations such that all the intermediate words are also $f$-free. This new setting is here studied giving a full characterization of the tilde-isometric words in terms of overlaps with errors.
\end{abstract}

\section{Introduction}\label{s-intro} 
The notion of edit distance  plays a crucial role in defining combinatorial properties of families of words as well as in designing many classical string algorithms that find applications in natural language processing, bioinformatics and, in general, in information retrieval problems. The edit distance is a string metric that quantifies how much two strings differ from each other and it is based on counting the minimum number of edit operations required to transform one string into the other one.

Changing the set of  the operations will give rise to   a different definition of edit distance.  The operations of insertion, deletion and replacement of a character in the string characterize the Levenshtein distance which is probably the most widely known (cf. \cite{lev66}). On the other hand, the most basic one is the Hamming distance which applies only to pair of strings having equal length and counts the positions where they have a mismatch; this corresponds to the restriction of using only the replacement operation. 
The Hamming distance finds a direct application in detecting and correcting errors in strings and it plays a main role   in the algorithms for string matching with mismatches (see \cite{galil90}).

The notion of isometric word  combines the notion of edit distance  with the property that a word does not appear as factor in other words.  Note that this property is important in combinatorics as well as in the investigation on similarities of DNA sequences, where the avoided factor is referred to as an absent word \cite{BMR96,CMR20,CCFMP,EGMRS07}. 
 
 Isometric words based on Hamming distance were first introduced 
 in \cite{IlicK12}
as special binary strings that never appear as factors in some string transformations. 
A string is $f$-{free} if it does not contain $f$ as factor. A word $f$ is isometric if for any pair of $f$-free words $u$ and $v$, there exists a sequence of symbol replacement operations that transforms $u$ into $v$, where all the intermediate words are also $f$-free.
In this original definition, 
isometric words are strictly related to and motivated by  the special isometric subgraphs of the hypercubes, called generalized Fibonacci cubes.
The hypercube 
$Q_n$ is a graph whose vertices are the (binary) words of length $n$, and two vertices are adjacent when the corresponding words differ in exactly one symbol. Therefore, the distance between two vertices  in $Q_n$ is the Hamming distance of the corresponding vertex-words. The generalized Fibonacci cube $Q_n(f)$ is  the subgraph of $Q_n$ which contains only vertices that are $f$-free where $f$ is an isometric word. This implies that the distances of the vertices in $Q_n(f)$ are the same as calculated in the whole $Q_n$, that is $Q_n(f)$ is isometric to $Q_n$. Fibonacci cubes have been introduced by Hsu in \cite{Hsu93} and correspond to the case with $f=11$. 
In \cite{IlicK12,KlavzarS12,Wei17,WeiYZ19,Wei16} the structure of non-isometric words for alphabets of size 2 and Hamming distance is completely characterized and related to particular properties on their overlaps. The more general case of alphabets of size greater than 2 and Lee distance is studied in \cite{AFMWords21,AnselmoFM22,MMM22,2024-AFM-Computa,AFMIctcs23}.  
Using these characterizations, in \cite{BealC22} some linear-time algorithms are given to check whether a binary word is Hamming isometric and, for quaternary words, if it is Lee isometric. 
 These algorithms were extended to provide further information on  non-isometric words, still keeping linear complexity in \cite{MMM22}.
Binary Hamming isometric two-dimensional words have been also studied in \cite{AGMS-DCFS20}.

Many challenging problems in correcting errors in strings come from computational biology. Among the chromosomal operations on DNA sequences, in gene mutations and duplication, it seems natural to consider  the {\em swap} operation, consisting in exchanging two adjacent symbols. 
The  Damerau-Levenshtein distance  adds also the swap to all edit operations. In \cite{Wagner75}, Wagner proves that computing the edit distance with insertion, deletion, replacement, and swap, is polynomially solvable in some restriction of the problem. The 
swap-matching problem has been considered in \cite{AmirCHLP03,FaroP18,EFMMPS23,AnselmoCFGMM23,ACFGMMDcfs23,CastGiamFlo23}, and related algorithms are given in \cite{AmirEP06,DombbLPPT10}.

In this paper, we study the notion of binary isometric word using the edit distance based on swaps and mismatches. This distance will be here called {\em tilde-distance} by using the $\sim$ symbol  that somehow evokes the swap operation. 
The tilde-distance ${\rm dist}_\sim(u,v)$ of
  equal-length words $u$ and $v$ is the minimum number of replacement and swap operations needed to transform $u$ into $v$. 
 
 It turns out that the addition of the swap operation to the definition makes the situation 
more complex, {although interesting for  applications}. It is not a mere generalization of the Hamming case  since special situations arise. A swap operation in fact is equivalent to two replacements, but it counts as one when  computing the tilde-distance. Moreover, there could be different ways to transform $u$ into $v$ since particular triples of consecutive symbols can be managed, from left to right, either by a swap and then a replacement or by a replacement and then a swap. Furthermore
there is also the possibility of making two  swaps in consecutive positions that are in fact equivalent to two replacement operations; in this case there is a position whose symbol   is changed twice. This fact is the major difference  with the Hamming distance where all the  transformations from $u$ to $v$ may differ only in  the order in which the replacemnt operations are applyed to the  positions to be changed. This difference will spawn new scenario when dealing with tilde-isometricity.

The definition of {\em  tilde-isometric} word comes in a very natural way. A word $f$ is tilde-isometric if for any pair of equal-length words $u$ and $v$ that are $f$-free, there is a tilde-transformation from $u$ to $v$ that uses exactly dist$_\sim(u,v)$ replacement and swap operations and such that all the intermediate words still avoid $f$. As expected, the notions is not a mere generalization of the classic case. In fact we exibit 
  some examples of tilde-isometric words that are not Hamming isometric and vice versa.

  The main result of this paper consists in a complete characterization of the words that are not  tilde-isometric   in terms of special configurations in their overlap. Note that, in order to prove that a given string $f$ is not tilde-isometric one should exhibit a pair of $f$-free words $(\tilde\alpha, \tilde\beta)$ such that any tilde-transformation from $\tilde\alpha$ to $\tilde\beta$ of length dist$_\sim(\tilde\alpha, \tilde\beta)$ comes through words that contain $f$. Such a pair is called pair of {\em tilde-witnesses} for $f$. The proof or our main theorem  gives also  an explicit construction of the tilde-witnesses in all possible cases.

\medskip 
The paper is organized as follows. In a short  preliminary section  we first fix some basic notations on words used through the paper and then we  report known results on (Hamming) isometric words. Then two main sections follow. The first one starts by introducing the tilde-distance together with examples and the notion of tilde-transformations. This leads to the formal definition of tilde-isometric words and some examples of them. Then, we study the overlap with tilde-errors and finally we state our main result that consists in the fully characterization of tilde-isometric words in terms of special properties of their overlaps. The last section is completely dedicated to the proof of our main characterization theorem. The two directions of the proofs are given in two separate subsections respectively and are preceded by some technical lemmas. Finally, a partial and preliminary version of the results in this paper can be found in \cite{AnselmoCFGMM23}.
\section{Preliminaries}\label{s:preli}

Let $\Sigma$ be a finite alphabet.
A word (or string) $w$ of length $|w|=n$, is $w=a_1a_2\cdots a_n$, 
where $a_1, a_2, \ldots, a_n$ are symbols in $\Sigma$. The set of all words over $\Sigma$ is denoted $\Sigma^*$ and the set of all words over $\Sigma$ of length $n$ is denoted $\Sigma^n$. Finally, $\epsilon$ denotes the {\em empty word} and $\Sigma^+=\Sigma^* - \{\epsilon\}.$
For any word $w=a_1a_2\cdots a_n$, the {\em reverse} of $w$ is the word $w^{rev}=a_na_{n-1}\cdots a_1$. If $x \in \{0, 1\}$, we denote by $\overline{x}$ the opposite of $x$, i.e $\overline{x}=1$ if $x=0$ and vice versa. 
Then we define {\em complement} of $w$ the word $\overline{w}=\overline{a}_1\overline{a}_2\cdots \overline{a}_n$.

Let 
  $w[i]$ denote the symbol of $w$ in position $i$, i.e. $w[i]=a_i$.
  If $f=w[i .. j] = a_i \cdots a_j$,
   for $1\leq i\leq j\leq n$, we say that $f$ is a \emph{factor} of $w$ that occurs in the interval $[i .. j]$, or equivalently, at position $i$. 
  The \emph{prefix}  (resp. \emph{suffix}) of $w$ of length $l$, with $1 \leq \ell \leq n-1$ is ${\rm pre}_\ell(w) = w[1 .. \ell]$  (resp. ${\rm suf}_\ell (w) = w[n-\ell+1 .. n]$).
  When ${\rm pre}_\ell(w) = {\rm suf}_\ell (w)=u$ then $u$ is here referred to as an \emph{overlap} of $w$ of length $\ell$; 
  it is also called border, or bifix. A word $w$ is said {\em $f$-free} if $w$ does not contain $f$ as a factor.

An {\em edit operation} is a function $O: \Sigma^* \to \Sigma^*$ that transforms a word into another one. Among the most common edit operations there are the insertion, the  deletion or the replacement of a character and the swap of two adjacent characters. 
Let $\mathcal{O}$ be a {\em set of edit operations}. The {\em edit distance} of two words $u, v \in\Sigma^*$ 
is the minimum number of edit operations in $\mathcal{O}$ needed to transform $u$ into $v$. 

A well-known edit distance is the  {\em Hamming distance} which is defined over a binary alphabet and uses only the replacement operation. Based on the Hamming distance, the definition of {\em good} words
is introduced in \cite{IlickKR2012} and

put in connection with isometric subgraphs of the hypercubes in \cite{KlavzarS12}. In this paper we refer to this definition of isometric as {\em Ham-isometric}. More specifically, 
a word $f$ is {\em Ham-isometric} if for any pair of $f$-free words $u$ and $v$, there exists a sequence of replacement operations that applied to $u$ transform $u$ into $v$ keeping all the intermediate words $f$-free. 
In \cite{WeiYZ19}, it is given a characterization of   Ham-isometric words based on a property of their overlaps as follows.
Here, a word $w$ has a $2$-error overlap  if  there exists $1<\ell<|w|$ such that
${\rm pre}_\ell (w)$ and ${\rm suf}_\ell (w)$ differ in exacly 2 positions (i.e., they have Hamming distance equal to 2). 
\begin{proposition}\label{p-Hamiso}
   A word $f$ is not Ham-isometric if and only if $f$ has a 2-error overlap.
\end{proposition}
For example the binary word $111000$ is not Ham-isometric (take the prefix $11$ and the suffix $00$) while it is easy to verify that the word $1010$ does not have any 2-error overlap and therefore it is Ham-isometric. 

\section{Tilde-distance and Tilde-isometric words}\label{s-tildeiso}
In this section, we introduce the tilde-distance as the edit distance based on replacements and swaps. Then, using the tilde-distance,  we define the tilde-isometric words as a generalization of the  (classical) Ham-isometric words. Very interesting situations and examples arise since the possibility of applying swap operations makes very different scenarios. At the end of the section we state the main   theorem of the paper that provides a complete characterization of tilde-isometric words in terms of special blocks contained in their overlaps.
\subsection{Tilde-distance and tilde-transformations}\label{s-tildedist}
We start by formally introducing   the operations of replacements and swaps; they are the base for the definition of tilde-distance. Then, we point out  some interesting facts regarding the tilde-transformations between words considered in the definition of distance.

\begin{definition}
    Let $\Sigma$ be a finite alphabet and $w=a_1a_2\ldots a_n$ a word over $\Sigma$.
    The {\em replacement operation} (or {\em replacement}, for short) on $w$ at position $i$ with $x\in \Sigma$, $x\neq a_i$, is defined by
    $$R_{i,x}(a_1a_2\ldots a_{i-1}a_ia_{i+1}\ldots a_n)=a_1a_2\ldots a_{i-1} x a_{i+1}\ldots a_n.$$
    The {\em swap operation} (or {\em swap}, for short) on $w$ at position $i$ consists in exchanging  characters at positions $i$ and $i+1$, provided that they are different, $a_i \neq a_{i+1}$, 
    $$S_i(a_1a_2\ldots a_ia_{i+1}\ldots a_n)=a_1a_2\ldots a_{i+1}a_i \ldots a_n.$$
     \noindent 
\end{definition}

When the alphabet $\Sigma=\{0,1\}$ there is only one possible replacement at a given position $i$, so we  write $R_i(w)$ instead of $R_{i,x}(w)$. 
In the following we always assume $\Sigma=\{0,1\}$.

 Remark that a single swap operation is a sort of shortcut since it has the same effect of two replacements on two adjacent positions respectively.
Moreover note that both replacement and swap operations do not change the length of the word on which they are applied.

The edit distance based on swap and replacement operations will be called {\em tilde-distance} and denoted $\tdist$.

\begin{definition}
\label{def:d-tilde-distance}
Let $u, v \in \Sigma^n$. The {\em tilde-distance} ${\rm dist}_\sim(u,v)$ between $u$ and $v$ is the minimum number of replacements and swaps needed to transform $u$ into $v$.
\end{definition}

In this paper, 
we focus on the sequences of operations that transform a word  into another one. We give first a formal definition of a transformation in the setting of the tilde-distance.

\begin{definition}\label{def:tilde-transformation}
  Let $u, v \in \Sigma^n$ be words of equal length and ${\rm dist}_\sim(u,v)=d$. A \emph{tilde-transformation} $\tau$ from $u$ to $v$ is a sequence of $d+1$ words $(w_0, w_1, \ldots, w_d)$ such that $w_0=u$, $w_d=v$, and for any $k=0, 1, \ldots ,d-1$, ${\rm dist}_\sim(w_k,w_{k+1})=1$. Moreover, let $f \in \Sigma^+$, $\tau$ is {\em $f$-free} if for any $i = 0,1, \ldots ,d$, the word $w_i$ is $f$-free. \end{definition}

A tilde-transformation $(w_0, w_1, \ldots, w_d)$ from $u$ to $v$ is associated to a sequence of $d$ operations $(O_{i_1}, O_{i_2},\ldots O_{i_d})$  such that, for any $k=1, \ldots ,d$,  $O_{i_k} \in \{R_{i_k},S_{i_k}\}$ and $w_{k}=O_{i_k}(w_{k-1})$; it can be represented as follows: 
 $$u=w_0\xrightarrow{O_{i_1}}w_1 \xrightarrow{O_{i_2}} \cdots \xrightarrow{O_{i_d}}w_d=v.$$
Therefore, an equivalent way to refer to a tilde-transformation is to give the sequence $(O_{i_1}, O_{i_2},\ldots O_{i_d})$ of its operations. 
When considering a tilde-transformation of $u$ into $v$, it is helpful to easily detect the positions where they differ. For this, it is useful to write $u$ and $v$ aligned one above the other. So, in what follows, $u\choose v$ will denote an alignment of two words $u$ and $v$ of equal length $n$, meaning that for all $1\leq k\leq n$, $u[k]$ is aligned on $v[k]$.
 When an alignment 
 $u \choose v$  is considered, a factor of length $m$ of the alignment,  denoted by $x\choose y$, is called a \emph{block} of $u\choose v$ and we say that it occurs in position $i$ if $x, y$ occur in $i$ as factors of $u$ and $v$, respectively. For instance, the block $100 \choose 011$ occurs at position $3$ in $001001 \choose 010110$.

Observe that the number of the operations in a tilde-transformation should correspond to the distance between the two words, but we do not put any restrictions on the kind of operations involved. In particular we also allow a position $i$ to be changed twice by a pair of swaps $S_{i-1}$ and $S_{i}$, for some position $i$ (note that this was not permitted in the previous paper \cite{AnselmoCFGMM23}).
Further, observe that a sequence of swaps at consecutive positions $S_i$, $S_{i+1}, \dots, S_{i+k}$, with $k\geq 2$,  cannot be part of a tilde-transformation since the same result can be obtained by two replacements $R_i$ and $R_{i+k+1}$ in a shorter sequence of operations.

Another important issue is that, given two words $u$ and $v$, there are many tilde-transformations between them; such transformations may differ in the type of operations we choose and also in the order in which we perform such operations. Pay attention to the fact that you cannot always change the order of swaps. This is a major difference with the case of Hamming distance where only replacements are allowed and therefore all possible transformations can be obtained by exchanging the order of the operation application. The following significant cases explain this fact.

\begin{example}\label{ex:101_010}
    Let ${u\choose v}={101\choose 010}$. Then $(S_1, R_3)$, $(R_3, S_1)$, $(R_1, S_2)$ and $(S_2, R_1)$ are all the different tilde-transformations from $u$ to $v$. 
\end{example}
\begin{example}\label{ex:100_001}
    Let ${u\choose v}={100\choose 001}$. Then $(R_1, R_3)$ $(R_3, R_1)$ and $(S_1, S_2)$ are three tilde-transformations from $u$ to $v$. Note that the operations $S_1$ and $S_2$ do not commute, since $S_2$ cannot be executed before $S_1$. Therefore the pair $(u,v)$ admits three different tilde-transformations. Note also that tilde-transformation $(S_1, S_2)$ flips twice the bit in position $2$.
\end{example}

Define the sets of blocks
$\mathcal{B}_0=\{{10 \choose 01}, {01\choose 10}\}$, 
$\mathcal{B}_1=\{{101 \choose 010}, {010\choose 101}\}$ and $\mathcal{B}_2=\{{100 \choose 001}, {110\choose 011}{001\choose 100}, {011\choose 110}\}$, 
obtained by applying complement, reverse and exchange of rows to blocks $10\choose 01$, $101\choose 010$ and $100 \choose 001$, 
respectively. In what follows we will only consider one block in $\mathcal{B}_0$, $\mathcal{B}_1$ or in $\mathcal{B}_2$, respectively, as representative of its class, since the results on the other blocks can be easily  inferred from them.  

\begin{remark}\label{r-unici-blocchi-ambigui-dist2}
 Let $u,v\in\Sigma^*$. As a generalization of Example \ref{ex:101_010}, note that, if $u[i]=u[i+2]$, each time $R_i$ and $S_{i+1}$ belong to a tilde-transformation $\tau$, the sequence obtained from $\tau$ by substituting $R_i$ and $S_{i+1}$ with $S_i$ and $R_{i+2}$ is still a tilde-transformation, and vice versa. Similarly, as a generalization of Example \ref{ex:100_001}, the operations $S_i$ and $S_{i+1}$ can be substituted with $R_i$ and $R_{i+2}$ in a tilde-transformation; the vice versa is possible only when the block occurring in $u \choose v$ at position $i$ belongs to $\mathcal{B}_2$. In particular, when ${\rm dist}_\sim(u,v)=2$, 
 exactly two tilde-transformations from $u$ to $v$ exist, except in those cases in which a block in 
$\mathcal{B}_1$ occurs in $u \choose v$ yielding four tilde-transformations from $u$ to $v$ (cf. Example \ref{ex:101_010}), and when a block in $\mathcal{B}_2$ occurs in 
 $u \choose v$ and produces three tilde-transformations from $u$ to $v$ (cf. Example \ref{ex:100_001}). 
\end{remark}

\subsection{Tilde-isometric words}

Tilde-transformations between two words play an important role in the definition of tilde-isometric words, here defined in analogy with the Ham-isometric words, but on the base on the swap and mismatch distance.
We give here the formal definition.

\begin{definition}\label{d-tilde-iso}
	Let $f\in \Sigma^n$, with $n\geq 1$, $f$ is \emph{tilde-isometric} if for any pair of $f$-free words $u$ and $v$ of length $m\geq n$, there exists a tilde-transformation from $u$ to $v$ that is $f$-free.
    It is \emph{tilde-non-isometric} if it is not tilde-isometric. 
\end{definition}

Note that, in order to prove that a word is tilde-non-isometric, it is sufficient to exhibit a pair $(u,v)$ of words contradicting  Definition \ref{d-tilde-iso}, that we call {\em pair of tilde-witnesses} for $f$. More formally we give the following definition.

\begin{definition}\label{d-witnesses}
    A pair $(u,v)$ of words in $\Sigma^m$ is a pair of \emph {tilde-witnesses} for $f$ 
    if: 
    
    1. ${\rm dist}_{\sim}(u,v) \geq 2$
        
    2. $u$ and $v$ are $f$-free
    
    3. there exists no $f$-free tilde-transformation from $u$ to $v$.
\end{definition}

\begin{example}\label{e-1010}
Let $f=1010$, then $f$ is tilde-non-isometric because the pair $(u, v)=(11000, 10110)$ is a pair of tilde-witnesses for $f$. In fact,  observe that there are  two tilde-transformations from $u$ to $v$, namely:
$$11000 \xrightarrow{S_2} \underline{1010}0 \xrightarrow{R_4} 10110 \hspace{1cm}11000 \xrightarrow{R_4} 1\underline{1010} \xrightarrow{S_2} 10110.$$ In both tilde-transformations $1010$ appears as factor after the first step, as evidenced by the underlined characters, i.e. they are not $f$-free.
\end{example}
\begin{example}
    Let $f=100011$ and consider the words $u$ and $v$ given by  ${u\choose v}={100101011\choose 100010011}$. Four tilde-transformations, exist from $u$ to $v$,  namely:
    $$100101011 \xrightarrow{R_4} 100001011 \xrightarrow{S_5} 100010011;
    100101011 \xrightarrow{S_5} 100110011 \xrightarrow{R_4} 100010011,$$  
    $$100101011 \xrightarrow{S_4} \underline{100011}011 \xrightarrow{R_6} 100010011;
    100101011 \xrightarrow{R_6} 100\underline{100011} \xrightarrow{S_4} 100010011.$$ 
    Note that the first and the second tilde-transformations are $f$-free whereas the third and the fourth are not. This proves that $(u,v)$ is not a pair of tilde-witnesses. Note that $u \choose v$ contains the block $101 \choose 010$ at position $4$.
\end{example}

As a matter of fact, a pair of tilde-witnesses for the word $f=100011$ of previous example does not exist, i.e. $f$ is tilde-isometric. This can be   proved by exploiting the  the main theorem of this paper that consists of a characterization of tilde-isometric words in terms  of a property of their overlaps.
From now on, we study isometric binary words starting with  $1$, in view of the following lemma whose  proof can be easily inferred by the definition.

\begin{lemma}\label{l-rev-comp}
	Let $f\in\Sigma^n$. The following statements are equivalent:
 \begin{enumerate}
     \item $f$ is tilde-isometric
     \item $f^{rev}$ is tilde-isometric 
     \item $\overline{f}$  is tilde-isometric.
 \end{enumerate}
\end{lemma}


Although the tilde-distance is more general than the Hamming distance,  
the sets of Ham-isometric and tilde-isometric words are incomparable, as stated in the following proposition.

\begin{proposition}\label{p-compareHam}
There exists a word which is tilde-isometric but Ham-non-isometric, and
a word which is tilde-non-isometric, but Ham-isometric.
\end{proposition}

\begin{proof}
The word $f=111000$ is tilde-isometric (cf. Example  \ref{ex:1_n_o_m}), but $f$ is Ham-non-isometric. In fact, $f$ has a $2$-tilde-error overlap (with respect to the Hamming distance) with shift 4 and then it
is Ham-non-isometric
by Proposition \ref{p-Hamiso}. 
	\\	
Conversely, $f'=1010$ is tilde-non-isometric (see Example \ref{e-1010}), but Ham-isometric by Proposition \ref{p-Hamiso}. \end{proof}

\subsection{Tilde-error overlaps}\label{s-iso-e-2eo}
The characterization of Ham-isometric words given in \cite{WeiYZ19} and here reported as Proposition \ref{p-Hamiso}, uses the notion of $2$-error overlap. In this section we introduce the corresponding definition that refers to the tilde-distance. Tilde-error overlaps will have a main role in the characterization of tilde-isometric words but the presence of swap operations will force us to handle them with care. 

 \begin{definition}
	Let $f\in \Sigma^n$.
	Then, $f$ has a {\em $q$-tilde-error overlap} of length $\ell$ and shift $r=n-\ell$,  with
	$1 \leq \ell \leq n-1$ and $0\leq q \leq \ell$,
	if
	${\rm dist}_{\sim} ({\rm pre}_\ell(f), {\rm suf}_\ell(f))=q$.

\end{definition}

\begin{example}\label{ex:2eo}
    The word $f=1101110101101$ has a $2$-tilde-error overlap of length $6$ and shift $7$. Indeed, ${\rm pre}_6(f)=110111$, ${\rm suf}_6(f)=101101$ and ${\rm dist}_{\sim} (110111, 101101)=2$.
\end{example}

For our proofs, given a word $f$ with  a $q$-tilde-error overlap of length $\ell$, we will study the tilde-transformations $\tau$ from ${\rm pre}_\ell(f)$ to ${\rm suf}_\ell(f)$ with $q$ operations. For this reason we need to refer to the alignment of the two strings ${\rm pre}_\ell(f)$ to ${\rm suf}_\ell(f)$. Furthermore, in some technical proofs it will sometimes be relevant to consider the bits adjacent to a tilde-error overlap of a word. For this reason we introduce the following notation.
Let $f$ be a word in $\{0,1\}^*$ and $\$$ be a symbol different from $0,1$, here used as delimiter of a word  that ``matches" any symbol of the word. Consider $f$ with its delimiters $\$f\$$. 
A $q$-tilde-error overlap of length $\ell$ is denoted by  $\$xa\choose by\$$ where $xa$, $by$ are a prefix and a suffix, respectively, of $f$, $a,b\in \Sigma$, $x,y\in \Sigma^*$ with $|x|=|y|=\ell$, and $\tdist(\$xa, by\$)=\tdist(x,y)=q$. This notation makes evident the fact that in $f$ the prefix $x$ is followed by $a$ and the suffix $y$ is preceded by $b$. Moreover, a  $q$-tilde-error overlap $\$xa\choose by\$$ is sometimes factorized into blocks to highlight the significant part. For example, the $2$-tilde-error overlap of Example \ref{ex:2eo} is denoted by ${\$1 \choose 01} {1011 \choose 0110} {10 \choose 0\$}$ because ${\rm dist}_{\sim} (110111, 101101)=2={\rm dist}_{\sim} (1011, 0110)$.

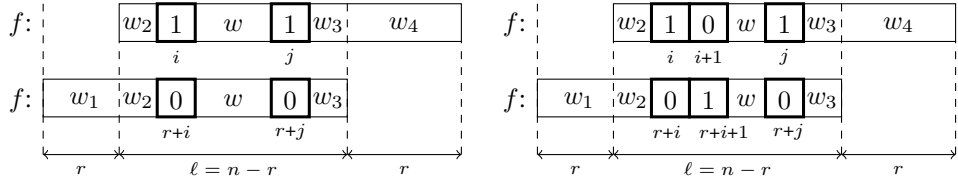
\begin{figure}[ht]
    \begin{center}
        \begin{tikzpicture}
		\node at (0,2.32){$f$:};
		\node at (2.05,1.9){\scriptsize $i$};
		\node at (3.55,1.9){\scriptsize $j$};
		\node at (1.55,2.32){$w_2$};
		\node at (2.05,2.32){$1$};
		\node at (2.8,2.32){$w$};
		\node at (3.55,2.32){$1$};
		\node at (4.05,2.32){$w_3$};
		\node at (5.05,2.32){$w_4$};
		\node at (0,1.32){$f$:};
		\node at (2.05,0.9){\scriptsize $r$+$i$};
		\node at (3.55,0.9){\scriptsize $r$+$j$};
		\node at (0.85,1.32){$w_1$};
		\node at (1.55,1.32){$w_2$};
		\node at (2.05,1.32){$0$};
		\node at (2.8,1.32){$w$};
		\node at (3.55,1.32){$0$};
		\node at (4.05,1.32){$w_3$};
		\node at (0.8,0.4){\scriptsize $r$};
		\node at (2.8,0.4){\scriptsize $\ell=n-r$};
		\node at (5.05,0.4){\scriptsize $r$};
		\node at (6.5,2.32){$f$:};
		\node at (8.55,1.9){\scriptsize $i$};
		\node at (9.05,1.9){\scriptsize $i$+$1$};
		\node at (10.05,1.9){\scriptsize $j$};
		\node at (8.05,2.32){$w_2$};
		\node at (8.55,2.32){$1$};
		\node at (9.05,2.32){$0$};
		\node at (9.55,2.32){$w$};
		\node at (10.05,2.32){$1$};
		\node at (10.55,2.32){$w_3$};
		\node at (11.55,2.32){$w_4$};
		\node at (6.5,1.32){$f$:};
		\node at (8.5,0.9){\scriptsize $r$+$i$};
		\node at (9.25,0.9){\scriptsize $r$+$i$+$1$};
		\node at (10.1,0.9){\scriptsize $r$+$j$};
		\node at (7.35,1.32){$w_1$};
		\node at (8.05,1.32){$w_2$};
		\node at (8.55,1.32){$0$};
		\node at (9.05,1.32){$1$};
		\node at (9.55,1.32){$w$};
		\node at (10.05,1.32){$0$};
		\node at (10.55,1.32){$w_3$};
		\node at (7.3,0.4){\scriptsize $r$};
		\node at (9.3,0.4){\scriptsize $\ell=n-r$};
		\node at (11.55,0.4){\scriptsize $r$};
		\draw (1.3,2.1) rectangle (5.8,2.6);
		\draw (0.3,1.1) rectangle (4.3,1.6);
		\draw (7.8,2.1) rectangle (12.3,2.6); 
		\draw (6.8,1.1) rectangle (10.8,1.6);
		\draw [line width=0.4mm] (1.8,2.1) rectangle (2.3,2.6);
		\draw [line width=0.4mm] (3.3,2.1) rectangle (3.8,2.6);
		\draw [line width=0.4mm] (1.8,1.1) rectangle (2.3,1.6);
		\draw [line width=0.4mm] (3.3,1.1) rectangle (3.8,1.6);
		\draw [line width=0.4mm] (8.3,2.1) rectangle (8.8,2.6);
		\draw [line width=0.4mm] (8.8,2.1) rectangle (9.3,2.6);
		\draw [line width=0.4mm] (9.8,2.1) rectangle (10.3,2.6);
		\draw [line width=0.4mm] (8.3,1.1) rectangle (8.8,1.6);
		\draw [line width=0.4mm] (8.8,1.1) rectangle (9.3,1.6);
		\draw [line width=0.4mm] (9.8,1.1) rectangle (10.3,1.6);
		\draw [dashed] (6.8, 0.6) -- (6.8,2.6);
		\draw [dashed] (7.8, 0.6) -- (7.8,2.1);
		\draw [dashed] (10.8, 0.6) -- (10.8, 2.6);
		\draw [dashed] (12.3, 0.6) -- (12.3,2.1);
		\draw [dashed] (0.3, 0.6) -- (0.3,2.6);
		\draw [dashed] (1.3, 0.6) -- (1.3,2.1);
		\draw [dashed] (4.3, 0.6) -- (4.3, 2.6);
		\draw [dashed] (5.8, 0.6) -- (5.8,2.1);
		\draw [<->] (0.3,0.6) -- (1.3,0.6);
		\draw [<->] (1.3,0.6) -- (4.3,0.6);
		\draw [<->] (4.3,0.6) -- (5.8,0.6);
		\draw [<->] (6.8,0.6) -- (7.8,0.6);
		\draw [<->] (7.8,0.6) -- (10.8,0.6);
		\draw [<->] (10.8,0.6) -- (12.3,0.6);
	\end{tikzpicture}
    \end{center}
    \caption{A word $f$ and its $2$-tilde-error overlap of shift $r$ and length $\ell=n-r$, with tilde-transformation $(O_i, O_j)=(R_i,R_j)$ (left), and $(O_i, O_j)=(S_i,R_j)$ (right)}. 
    \label{f-w-overlap}
\end{figure}

In the sequel, we will be interested in the specific case of $1$-tilde-error overlap where the single error in the alignment is a swap and  in all the cases of $2$-tilde-error overlaps.  
Consider a $2$-tilde-error overlap of $f$ of shift $r$, length $\ell=n-r$, and let $(O_i,O_j)$, $1\leq i<j \leq\ell$, be a tilde-transformation from ${\rm pre}_{\ell}(f)$ to ${\rm suf}_{\ell}(f)$.
Observe that the positions in 
${\rm pre}_{\ell}(f)$ modified by $O_i$ 
are either $i$ or both $i$ and $i+1$, following that $O_i$ is a replacement or a swap.
Hence, the positions modified by $O_i$ and $O_j$ may be 2, 3 or 4.
Fig. \ref{f-w-overlap} shows a word $f$ with its $2$-tilde-error overlap of shift $r$ and length $\ell=n-r$. With our notation, the $2$-tilde-error overlap is  $\$w_21w1w_3a \choose bw_20w0w_3\$$ in the figure on the left
and
$\$w_210w1w_3a \choose bw_201w0w_3\$$ in the figure on the right,
where $a$ is 
the last letter of $w_1$ and $b$ 
the first letter of $w_4$. 
A tilde-transformation from ${\rm pre}_{\ell}(f)$ to ${\rm suf}_{\ell}(f)$ is given by $(O_i, O_j)=(R_i, R_j)$ in the first case
and by $(O_i, O_j)=(S_i, R_j)$ in the second case.
We say that a $2$-tilde-error overlap has {\em non-adjacent errors} when 
 there is at least one character interleaving 
the positions modified by $O_i$ and those modified by $O_j$.

The $2$-tilde-error overlap ${\$x \choose ax }{100 \choose 001}{yb \choose y\$}$ is also considered as having non-adjacent errors because it admits the tilde-transformation $(O_i, O_j)=(R_i,R_{i+2})$, despite it has also the other $(O_i, O_j)=(S_i,S_{i+1})$. 

In all the other cases,
we say that the $2$-tilde-error overlap has {\em adjacent errors}.  

\section{Characterization of tilde-isometric words}
Let us state the main result of the paper that consists in the characterization of tilde-isometric words in terms of special configurations in their overlap. The proof of the theorem is quite complicate and is given separately in the next section.

Recall that in this paper a block in
 $\mathcal{B}_0$, 
 $\mathcal{B}_1$ or $\mathcal{B}_2$ is
 taken as representative of  its class, since the results on the other blocks can  easily be inferred from it.

\begin{theorem}\label{t:characterization}
 A word $f \in \Sigma^n$ is tilde-non-isometric if and only if one of the following cases occurs (up to complement, reverse and inversion of rows):
\begin{itemize}
\setcounter{enumi}{-1}
    \item[$(C0)$]\label{case0} $f$ has a $1$-tilde-error overlap ${\$x\choose ax}{01\choose 10}{yb\choose y\$}$ with $x,y\in \Sigma^*$, $a, b\in \Sigma$\\
   \item[$(C1)$]\label{case1Teo}$f$ has a $2$-tilde-error overlap 
    with non-adjacent errors, 
different from   
    ${\$x \choose ax}{000\choose 101}{yb \choose y\$}$ with $x, y \in \Sigma^+$, $a,b\in \Sigma$
    \item[$(C2)$]\label{case2Teo} $f$ has a $2$-tilde-error overlap ${\$x\choose ax}{0101\choose 1010}{yb\choose y\$}$ 
    or ${\$x\choose ax}{0110\choose 1001}{yb\choose y\$}$
    with $x,y\in \Sigma^*$, $a,b\in \Sigma$

    \item[$(C3)$]\label{case3Teo} $f$ has a $2$-tilde-error overlap 
    ${\$x\choose ax} {{010\choose 101}} {yb\choose y\$}$ with $x,y\in \Sigma^*$, $a,b\in \Sigma$
     
    \item[$(C4)$]\label{case4Teo} $f$ has a $2$-tilde-error overlap 
    ${\$x\choose ax}{011\choose 100}{0\choose \$}$
    with $x\in \Sigma^*$, $a\in \Sigma$
    
    \item[$(C5)$]\label{case5Teo} $f$ has a $2$-tilde-error overlap ${\$ \choose 0}{00\choose 11}{1\choose \$}$
    
\end{itemize}
\end{theorem}
 
As an
 application of Theorem \ref{t:characterization} let us show the following example of a family of tilde-isometric words.

\begin{example}
All the words $f=1^n0^m$ for $n,m> 2$ are tilde-isometric. 
In fact, for $n,m> 2$, $f=1^n0^m$ has only two $2$-tilde-error overlaps and none of them fall into a case in the statement of Theorem \ref{t:characterization}.
The first one is the tilde-error overlap with shift $2$, ${\$1^{n-2} \choose 1 1^{n-2}}{11\choose 00}{0^{n-2}0 \choose 0^{n-2}\$}$,

the other one has shift $n+m-2$, and it is
${\$\choose 0}{11\choose 00}{1\choose \$}$.

Note that if $n=m=2$, $f=1100$ is tilde-non-isometric. In fact one can  verify that the pair of words ${u\choose v}={110100\choose 101010}$ is a pair of tilde-witnesses for $f$. As expected from Theorem \ref{t:characterization}, $f$ has only one $2$-tilde-error overlap ${\$ \choose 1}{11\choose 00}{0\choose \$}$, corresponding to case {\em ($C5$)}. 
Note also that words $f=1^n0^m$ with $n,m\geq 2$ are Hamming non-isometric. 
\end{example}\label{ex:1_n_o_m}

The following are two examples of words with a $2$-tilde-error overlap with adjacent errors. The first one is tilde-isometric, the second one is tilde-non-isometric.

 \begin{example}
     The word $f=010110000$ is tilde-isometric; indeed its unique $2$-tilde-error overlap has shift $5$ and length $4$, ${\$0 \choose 10}{101 \choose 000}{1 \choose \$}$. Note this is the case of non-adjacent errors but of the type prohibited by condition in (C1). 
 \end{example}

 \begin{example}
     The word $f=1011000$ is tilde-non-isometric; indeed it has the $2$-tilde-error overlap, of shift $4$ and length $3$, ${\$ \choose 1}{101 \choose 000}{1 \choose \$}$, that verifies (C1). Note that the pair $(u,v)=(10110011000, 10101001000)$ is a pair of tilde-witnesses for $f$.
 \end{example}
 
 \section{The proof of the Characterization Theorem}
This last section is very technical and contains the proof of Theorem \ref{t:characterization}. The cases (C0) up to (C5) list all the possible configurations we can find in an overlap of a tilde-non-isometric word. The proof walks carefully through all these cases to show that they cover all possible situations.  The two implications of the theorem are proved separately in two corresponding subsections and are preceded by some technical lemmas.

\subsection{Properties of tilde-witnesses with minimal distance}

When a word $f$ is tilde-non-isometric then there exists a pair of tilde-witnesses. We focus our attention on a pair $(u,v)$ of tilde-witnesses at minimal distance among all such pairs of a given length. With this constraint, any choice of the starting operation $O_i$ in a tilde-transformation from $u$ to $v$, will cause an occurrence of the factor $f$ in $O_i(u)$. More specifically, this occurrence of $f$ will cover at least one position modified by the operation  $O_i$ (i.e. position $i$ and/or $i+1$ in the case of a swap $S_i$).

We start with two technical lemmas.

\begin{lemma}\label{l-Is&It} 
Let $f\in\Sigma^n$ be a tilde-non-isometric word and $(u, v)$, with $u, v\in\Sigma^m$, be a pair of tilde-witnesses with minimal distance 
${\rm dist}_{\sim}(u,v)$
among all pairs of tilde-witnesses of length $m$. Let $\mathcal{O}=\{O_{i_1}, O_{i_2}, \dots , O_{i_d}\}$ be the set of the operations in a tilde-transformation from $u$ to $v$, that does not contain two swaps at two consecutive positions. Then, 
\begin{enumerate}
    \item\label{s_1} for any $i_j\in \{i_1, i_2, \dots , i_d\}$, $f$ is a factor of $O_{i_j}(u)$. Let $k_{i_j}$  denote the  starting position of such occurrence of $f$ in $O_{i_j}(u)$, and  $I_{i_j}=[k_{i_j}\ldots k_i+|f|-1]$ the interval where  $f$  occurs in $O_{i_j}(u)$.
    \item\label{s_2}  there exist $s,t \in \{i_1, i_2, \dots i_d\}$, with $s<t$, such that $I_s$ and $I_t$ share at least one position modified by $O_{s}$ and at least one position modified by $O_{t}$. Moreover, without loss of generality, $k_s<k_t$.
\end{enumerate}
\end{lemma}

\begin{proof} 
    Let $\mathcal{O}=\{O_{i_1}, O_{i_2}, \dots , O_{i_d}\}$ with $O_{i_j}\in \{ R_{i_j}, S_{i_j}  \}$ for any $j=1, 2, \dots, d$, and let 
$1\leq i_1< i_2< \dots  < i_d\leq m$.
Since the tilde-transformation does not contain two consecutive swaps, each $O_{i_j}$ can be applied to $u$. 
\\
If, for some $j\in \{1,\ldots, m\}$, $O_{i_j}(u)$ were $f$-free, then the pair $(O_{i_j}(u), v)$ would still be a pair of tilde-witnesses of length $m$, with $dist_\sim(O_{i_j}(u), v) < d$, against the hypothesis that $(u,v)$ are the tilde-witnesses of minimal distance. This proves statement (\ref{s_1}). 
\medskip
\\
Let $I_{i_j}$ be the interval where $f$ occurs in $O_{i_j}(u)$.
This interval contains at least one position modified by $O_{i_j}$, because $u$ is $f$-free;
let $o(j)$ denote the smallest position in $I_{i_j}$ modified by $O_{i_j}$.
Moreover, this occurrence of $f$ must disappear in a tilde-transformation from $u$ to $v$, because $v$ is $f$-free. Hence, $I_{i_j}$ contains a position modified by another operation in $\mathcal{O}$; let $p(j)$ denote the smallest such position.
Overall, there must exist $s,t \in \{i_1, i_2, \dots i_d\}$, such that  
$I_s$ contains at least one position modified by $O_{t}$ and 
$I_t$ contains at least one position modified by $O_{s}$.
To prove this, consider for any $j=1, \ldots, d$, the interval $I_{i_j}$ and the two positions 
$o(j)$ and $p(j)$
in $I_{i_j}$.
Note that $p(1)>o(1)$, whereas $p(d)<o(d)$.
Let $i_k$ be the smallest position in $\{i_1, i_2, \dots i_d\}$, such that 
$p(k)<o(k)$;
then $p(k-1)>o(k-1)$.
Since $I_{i_k}$ contains $o(k)$ and $p(k)$, with $p(k)< o(k)$, then $I_{i_k}$ also contains $o(k-1)$ 
because $p(k)\leq o(k-1)$. 
Then, $I_{i_k}$ and $I_{i_{k-1}}$ both contain $o(k-1)$ 
and $o(k)$ and can play the role of $I_s$ and $I_t$. 

Finally, without loss of generality we can suppose that $s<t$ and $k_s<k_t$. In fact, if this does not happen, exchange the roles of $u$ and $v$.\end{proof}

\begin{remark}\label{r:I_swap}
Suppose that $\mathcal{O}$ satisfy the hypotheses of Lemma \ref{l-Is&It}  and that contains a swap $S_i$. The swap operation $S_i$ modifies two positions, $i$ and $i+1$, and interval $I_i$ could contain just one of them. Then 
    $i-|f|+1 \leq k_i\leq i+1$.
\end{remark}
 
\begin{lemma}\label{l-blocchi-esclusi}
    Let $f$ be a tilde-non-isometric word and let $(u,v)$ be a pair of tilde-witnesses for $f$ of minimal distance
    ${\rm dist}_{\sim}(u,v)$. Let $O_s, O_t$ be as in Lemma \ref{l-Is&It},  then
    \begin{itemize}
    \item\label{l-101} If ($O_s=R_s$ and $O_t=S_{s+1}$) or ($O_s=S_s$ and $O_t=R_{s+2}$) then 
    \begin{center}
        ${u[s..s+2] \choose v[s..s+2]} \neq {101 \choose 010}$;
    \end{center}
        \item\label{l-100} If $O_s=R_s$ and $O_t=R_{s+2}$ then 
        \begin{center}
            ${u[s..s+2] \choose v[s..s+2]} \neq {100 \choose 001}$
        \end{center}
    \end{itemize} 
\end{lemma}

\begin{figure}[t]
    \begin{center}
	\begin{tikzpicture}
		\node[anchor=west] at (7,2.4){$u$};
		\node at (2.2,2.55){$1$};
		\node at (2.6,2.55){$0$};
		\node at (3,2.55){$1$};
		\node at (2.2,2.8){\tiny $s$};
		\node at (2.6,2.8){\tiny $s$+$1$};
		\node at (3.1,2.8){\tiny $s$+$2$};
		\node at (0,1.8){$f^3$};
            \node[anchor=west] at (7,1.8){$O^3$=$S_s$};
		\node at (2.2,1.95){$0$};
		\node at (2.6,1.95){$1$};
		\node at (3,1.95){$1$};
		\node at (3.6,1.95){$1$};
		\node at (1.05,1.95){\scriptsize $r$};
		\node at (3,2.2){\tiny $i$};
		\node at (3.9,2.2){\tiny $f^3[p$+$i]$};
            \node at (0,1.2){$f^2$};
		\node[anchor=west] at (7,1.2){$O^2$=$R_s$};
		\node at (2.2,1.35){$0$};
		\node at (2.6,1.35){$0$};
		\node at (3,1.35){$1$};
		\node at (1.6,1.35){$1$};
		\node at (3.6,1.35){$0$};
		\node at (1.5,1.6){\tiny $f^2[i]$};
		\node at (3.9,1){\tiny $f^3[r$+$i]$};
		\node at (0.55,1.35){\scriptsize $p$};
        \node at (0,0.6){$f^1$};		
        \node[anchor=west] at (7,0.6){$O^1$=$S_{s+1}$};
		\node at (2.2,0.75){$1$};
		\node at (2.6,0.75){$1$};
		\node at (3,0.75){$0$};
		\node at (1.6,0.75){$1$};
		\node at (3.3,0.4){\tiny $f^1[r$+$i]$};
		\node at (1.6,0.4){\tiny $f^1[p$+$i]$};
		\draw [|-|] (0.3,0.6) -- (4.3,0.6);
		\draw [|-|] (0.8,1.2) -- (5.05,1.2);
		\draw [|-|] (1.8,1.8) -- (5.8,1.8);
		\draw [|-|] (0.3,2.4) -- (6.8,2.4);
		\draw [dashed][<->] (0.3,1.2) -- (0.8,1.2);
		\draw [dashed][<->] (0.3,1.8) -- (1.8,1.8);
		\draw (3.45,1.2) rectangle (3.75,2.1);
		\draw (1.45,0.6) rectangle (1.75,1.5);

	\end{tikzpicture}
    \end{center}
    \caption{The representation of the three occurrences $f^1$, $f^2$, and $f^3$ of $f$ in $O^1(u)$, $O^2(u)$ and $O^3(u)$, respectively, when $u[s..s+2]=101$.}
    \label{f-101}	
\end{figure}

\begin{proof}
First, consider the case that
$O_s=R_s$, $O_t=S_{s+1}$. If ${u[s..s+2] \choose v[s..s+2]} = {101 \choose 010}$ we write $u=u_1101u_2$ and $v=v_1010v_2$, with $\vert u_1 \vert=s-1$. 

Then we can obtain another tilde-transformation from $u$ to $v$ by replacing $(R_s, S_{s+1})$ with $(S_s,R_{s+2})$ (see Remark \ref{r-unici-blocchi-ambigui-dist2}). By definition we have:
\begin{equation}\label{eq:blocchi}
\begin{split}
    R_s(u_1{\bf101}u_2)=u_1{\bf001}u_2, \,\,\,\,\, S_{s+1}(u_1{\bf101}u_2)=u_1{\bf110}u_2,\\  
    S_s(u_1{\bf101}u_2)=u_1{\bf011}u_2,\,\,\,\,\, R_{s+2}(u_1{\bf101}u_2)=u_1{\bf100}u_2
\end{split}
\end{equation}
The  proof works considering only $R_s$, $S_{s}$ and $S_{s+1}$.
By Lemma \ref{l-Is&It}, for each $O\in \{R_s, S_s, S_{s+1}\}$,  $O(u)$ has $f$ as factor. Therefore, one can sort the operations in the set $\{R_s, S_s, S_{s+1}\}$ according to the increasing order of the  positions where $f$ occurs as factor in $R_s(u)$, $S_{s+1}(u)$ and  $S_s(u)$, respectively. Hence, let $O^1, O^2, O^3$ be the sorted sequence of the operations, $f^1, f^2, f^3$ denote the relative occurrences of $f$ at positions $k_1 < k_2< k_3$ and $I_1$, $I_2$ and $I_3$ be the intervals of occurrence, respectively.
 Then, either ($O^3(u)[s+1]=a$ and $O^1(u) [s+1]=\overline{a}$) or ($O^3(u)[s+2]=a$ and $O^1(u) [s+2]=\overline{a}$), with $a \in \Sigma$. This fact is evident from Equation (\ref{eq:blocchi}) where the bits $s, s+1, s+2$ are in bold. In what follows, without loss of generality, we can consider $k_1=1$ and denote $r=k_3-1$, $p=k_2-1$ and $q =r-p$. 

We prove the statement in the case $O^1=S_{s+1}$, $O^2=R_{s}$ and $O^3=S_s$. The proofs for all other cases apply the same technique. In such a case $O^3(u)[s+2]=1$ and $O^1(u) [s+2]=0$ (see Fig. \ref{f-101}). Note that $I_2$ contains $s$, $s+1$ and $s+2$. 
Indeed, by Lemma \ref{l-Is&It}, $I_2$ contains $s$ and, at least, $s+1$, but if $I_2$ does not contain $s+2$ then $I_1$ does not contain $s+1$ (because $k_1 < k_2$) and ends in $s$. This is not possible because  $u$ would contain $f^1$, i.e. $u$ is not $f$-free. For the same reason, $k_3 \leq s+1$ because $u$ is $f$-free. Let $i=s+2-r$ be the position of $f^3$ corresponding to position $s+2$ of $u$ (see Fig.  \ref{f-101}); in other words, $i-2$, $i-1$ and $i$ are the positions of $f^3$ corresponding to $s$, $s+1$ and $s+2$, respectively.
Note that $f[i]=f^1[i]=f^2[i]=f^3[i]=1$ and $f[r+i]=f^3[r+i]=f^1[r+i]=f^2[r+i]=0$.
Observe that $q \neq 2$, and $q \neq 1$, indeed, $f[i]=f^3[i]=1$ and if either $q=1$ or $q=0$ then $f[i]=f^2[i]=0$, a contradiction.

Since $f^2$ occurs in $R_s(u)$, $f^2$ matches the corresponding positions of $u$ in $I_2$, unless for the one where the replacement is applied, that is $s$. More precisely,
for each $h \in I_2$ with $h \neq q+i-2$ we have $u[p+h]=f^2[h]$.

It follows that $u[p+i]=f^2[i]=1$, because $q \neq 2$. 

Moreover, since $f^1$ occurs in $S_{s+1}(u)$, $f^1$ matches the corresponding positions of $u$ in $I_1$, except for the positions involved in the swap $S_{s+1}$, i.e. $s+1$ and $s+2$. More precisely, for each $h \in I_1$,  with $h \neq r+i-1$ and $h \neq r+i$,
we have $u[h]=f^1[h]$.
It follows that $u[p+i]=f^1[p+i]=1$, because $q > 1$. 

 
Since $f^3$ occurs in $S_s(u)$, 
it matches the factor of $u$ occurring at $I_3$, except for the positions involved in the swap $S_s$, i.e. the ones corresponding to $s$ and $s+1$; in other words, 
 $u[r+h]=f^3[h]$, 
for each $h \in I_3$, with $h \neq i-2$ and $h \neq i-1$.

It follows that $u[r+p+i]=f^3[p+i]=1$ because $p > 0$.  
But $u[r+p+i]=f^2[r+i]=0$ because $p>0$. Then a contradiction follows.

The previous considerations can be extended to the other cases where $O^3(u)[s+2]=a$ and $O^1(u) [s+2]=\overline{a}$ getting the contradiction that $u[r+p+i]=a$ and $u[r+p+i]=\overline{a}$, with $i=s+2-r$. 
On the other hand, if $O^3(u)[s+1]=a$ and $O^1(u) [s+1]=\overline{a}$ the contradiction that $u[r+p+i]=a$ and $u[r+p+i]=\overline{a}$, can be similarly obtained for  $i=s+1-r$.

\smallskip

Now, consider the case $O_s=R_s$ and $O_t=R_{s+2}$. If ${u[s..s+2] \choose v[s..s+2]} = {100 \choose 001}$ we write $u=u_1100u_2$ and $v=v_1001v_2$, with $\vert u_1\vert =s-1$. Then we can can obtain another tilde-transformation from $u$ to $v$ by replacing  $(R_s, R_{s+2})$ with $(S_s,S_{s+1})$ (see Remark \ref{r-unici-blocchi-ambigui-dist2}). By definition:
\begin{equation}\label{eq:blocchi_2}
\begin{split}
    R_s(u_1{\bf100}u_2)=u_1{\bf000}u_2, \,\,\,\,\, R_{s+2}(u_1{\bf100}u_2)=u_1{\bf101}u_2,\\ 
S_s(u_1{\bf100}u_2)=u_1{\bf010}u_2,\,\,\,\,\,\,\,\,\,\,\,\,\,\,\,\,\,\,\,\,\,\,\,\,\,\,\,\,\,\,\,\,\,\,\,\,
\end{split}
\end{equation}

One can sort the set of operations $\{R_s, R_{s+2}, S_s\}$  as before mentioned and denote by $k_1 < k_2 < k_3$ the start positions of $f^1$, $f^2$ and $f^3$, and $I_1$, $I_2$ and $I_3$  the intervals of occurrence, respectively. From Equation (\ref{eq:blocchi_2}), observe that either ($O^3(u)[s+1]=a$ and $O^1(u)[s+1]=\overline{a}$) or ($O^3(u)[s+2]=a$ and $O^1(u)[s+2]=\overline{a}$). Without loss of generality we can consider $k_1=1$, $r=k_3-1$, $p=k_2-1$ and $q=r-p$. 

We prove the statement for the order $O^1=R_{s+2}$, $O^2=S_s$ and $O^3=R_s$ (see Fig. \ref{f-100}). Note that $O^3(u)[s+1]=1$ 
and $O^1(u)[s+1]=0$. By Lemma \ref{l-Is&It}, both $I_1$ and $I_3$ contain both $s$ and $s+2$. Therefore $I_2$ must contain $s$, $s+1$ and $s+2$.  Let $i=s+2-r$, $f[i]=f^1[i]=0=f^2[i]=f^3[i]$ and $f[r+i]=f^3[r+i]=1=f^1[r+i]=f^2[r+i]$.

\begin{figure}[t]
    \begin{center}
	\begin{tikzpicture}
		\node[anchor=west] at (7,2.4){$u$};
		\node at (2.2,2.55){$1$};
		\node at (2.6,2.55){$0$};
		\node at (3,2.55){$0$};
		\node at (2.2,2.8){\tiny $s$};
		\node at (2.6,2.8){\tiny $s$+$1$};
		\node at (3.1,2.8){\tiny $s$+$2$};
		\node at (0,1.8){$f^3$};
        \node[anchor=west] at (7,1.8){$O^3$=$R_s$};
		\node at (2.2,1.95){$0$};
		\node at (2.6,1.95){$0$};
		\node at (3,1.95){$0$};
		\node at (3.6,1.95){$0$};
		\node at (1.05,1.95){\scriptsize $r$};
		\node at (3,2.2){\tiny $i$};
		\node at (3.9,2.2){\tiny $f^3[p$+$i]$};
		\node at (0,1.2){$f^2$};
		\node[anchor=west] at (7,1.2){$O^2$=$S_s$};
		\node at (2.2,1.35){$0$};
		\node at (2.6,1.35){$1$};
		\node at (3,1.35){$0$};
		\node at (1.6,1.35){$0$};
		\node at (3.6,1.35){$1$};
		\node at (1.5,1.6){\tiny $f^2[i]$};
		\node at (3.9,1){\tiny $f^3[r$+$i]$};
		\node at (0.55,1.35){\scriptsize $p$};
		\node at (0,0.6){$f^1$};
		\node[anchor=west] at (7,0.6){$O^1$=$R_{s+2}$};
		\node at (2.2,0.75){$1$};
		\node at (2.6,0.75){$0$};
		\node at (3,0.75){$1$};
		\node at (1.6,0.75){$0$};
		\node at (3.3,0.4){\tiny $f^1[r$+$i]$};
		\node at (1.6,0.4){\tiny $f^1[p$+$i]$};
		\draw [|-|] (0.3,0.6) -- (4.3,0.6);
		\draw [|-|] (0.8,1.2) -- (5.05,1.2);
		\draw [|-|] (1.8,1.8) -- (5.8,1.8);
		\draw [|-|] (0.3,2.4) -- (6.8,2.4);
		\draw [dashed][<->] (0.3,1.2) -- (0.8,1.2);
		\draw [dashed][<->] (0.3,1.8) -- (1.8,1.8);
		\draw (3.45,1.2) rectangle (3.75,2.1);
		\draw (1.45,0.6) rectangle (1.75,1.5);
	\end{tikzpicture}
    \end{center}
    \caption{The representation of the three occurrences $f^1$, $f^2$, and $f^3$ of $f$ in $O^1(u)$, $O^2(u)$ and $O^3(u)$, respectively, when $u[s..s+2]=100$.}
    \label{f-100}	
\end{figure}
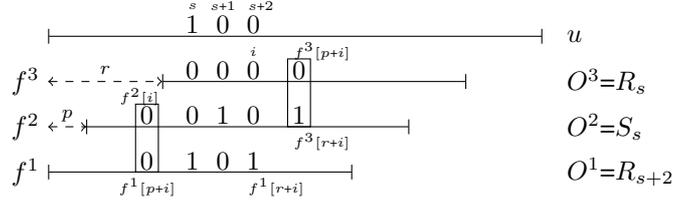

First, let us remark that $q \neq 2$ and $q \neq 1$. If $q=2$, one can easily verify that $f[h+1]=\overline{f[h]}$, for each $h>i$. By hypothesis $f^2[i+1]=1$ then $f^3[i+3]=1$ then $f^1[p+i+3]=1$, but $f^1[p+i+2]=1$ by hypothesis. A contradiction follows.
If $q=1$ then $f[i]=f^3[i]=0$ and $f[i]=f^2[i]=1$, a contradiction.

Since $f^2$ occurs in $S_s(u)$, $f^2$ matches the corresponding positions of $u$ in $I_2$ except for the positions involved in the swap $S_s$, i.e. $s$ and $s+1$. Hence, for each $h \in I_2$ with $h \neq q+i-2$ and $h \neq q+i-1$ we have $u[p+h]=f^2[h]$.

It follows that $u[p+i]=f^2[i]=1$, because $q \neq 2$ and $q \neq 1$. 
Moreover, since $f^1$ occurs in $R_{s+2}(u)$, $f^1$ matches the corresponding positions of $u$ in $I_1$ unless for $s+2$ where the replacement is applied. More formally, for each $h \in I_1$,  with $h \neq i$ we have $u[h]=f^1[h]$.
It follows that $u[p+i]=f^1[p+i]=1$, because $p > 0$.

On the other hand $f^1[p+i]=f^3[p+i]$. Since $f^3$ occurs in $R_s(u)$, $f^3$ matches the corresponding positions of $u$ in $I_3$ unless for $s$ where the replacement is applied. In other words, for each $h \in I_3$ with $h \neq i-2$ we have $u[r+h]=f^3[h]$. It follows that $u[r+p+i]=f^3[p+i]$ because $p >0$.

But $u[r+p+i]=f^2[r+i]=0$ because $p>0$. Then a contradiction follows.\end{proof}

 \subsection{The necessary condition}\label{s:viceversa}
  Once we have collected all the tools, we are now ready to prove  the  {\em only if}   direction of Theorem \ref{t:characterization}.

\begin{proposition}\label{p:necessary}
   If  $f \in\Sigma ^n$ is tilde-non-isometric then one among  (C0), (C1), (C2), (C3), (C4), (C5)  cases of Theorem \ref{t:characterization} occurs (up to complement, reverse and exchange of rows).
\end{proposition}

\begin{proof}
    Let $f$ be a tilde-non-isometric word, $(u,v)$ be a pair of tilde-witnesses for $f$ with minimal distance among the pairs of tilde-witnesses of minimal length. 
    Consider a tilde-transformation from $u$ to $v$ that does not contain two swaps at two consecutive positions. 
    
    Recall that if a tilde-transformation contains  swaps 
    $S_i$ and $S_{i+1}$ for some position $i$ in $u$ then they can be substituted by replacements $R_i$ and $R_{i+2}$ still keeping the total number of operations.
 Throughout the proof let $s, t, k_s, k_t, O_s, O_t, I_s, I_t$, with $k_s<k_t$, be as in Lemma \ref{l-Is&It}.
    
    Further suppose that $u[s]=c$ and $u[t]=d$ for some $c, d\in \Sigma$.

    First, consider the case where the characters modified by $O_s$ and $O_t$ are not adjacent. 
    
    If there are at least two characters of $u$ interleaving the characters modified by $O_s$ and the ones modified by $O_t$ then $f$ has a  $2$-tilde-error overlap with non-adjacent errors and any of its blocks containing all the error/ modified positions has length strictly greater than $3$, thus falling in case $(C1)$.

    Suppose now that there is only one character of $u$ interleaving the characters modified by $O_s$ and the ones modified by $O_t$. Then, four situations may occur.
    \begin{enumerate}
    	\item $O_s=R_s, O_t=R_t$ with $t=s+2$ and $u[s+1]=v[s+1]$. Let  $u[s+1]=e$ and $u=wcedz$ with $w, z\in \Sigma^*$.
 Then, $O_s(u)=w\overline{a}cbz$ and $O_t(u)=w ac\overline{b}z$ and, since $I_s$ and $I_t$ must share the positions $s$ and $t$ then $f$ has a $2$-tilde-error overlap with block $\overline{c}ed \choose ce\overline{d}$.     	
    	Applying Lemma \ref{l-blocchi-esclusi},  
${u[s..s+2] \choose v[s..s+2]} \neq {100 \choose 001}$ and then 
${\overline{c}ed \choose ce\overline{d}} \neq {000\choose 101}$,
 
     thus falling in case $(C1)$. 
   	\item $O_s=S_s, O_t=R_t$ with $t=s+3$  and $u[s+2]=v[s+2]$.
    	 Let  $u[s+2]=e$ and $u=wc\overline{c}edz$ with $w, z\in \Sigma^*$.
    	Then, $O_s(u)=w\overline{c}cedz$ and $O_t(u)=w c\overline{c}e \overline{d}z$.
    	If $I_s\cap I_t$ includes $s, s+1, s+2, s+3$ then $f$ has a  $2$-tilde-error overlap with non-adjacent error positions and any of its blocks containing the error positions has length strictly greater than $3$, thus falling in case $(C1)$.
    	If $I_s\cap I_t$ includes $s+1, s+2, s+3$, but not $s$, then $f$ 
    	has a  $2$-tilde-error overlap with non-adjacent error positions with block $\overline{c}e \overline{d} \choose ced$ where 
    	$\overline{c}e \overline{d}$ is a prefix of $f$. 
    	If $c=d$ 
    	then the block $000\choose 101$, may occur in the $2$-tilde-error overlap,      
     but the $2$-tilde-error overlap is not   
    ${\$x \choose ax}{000\choose 101}{yb \choose y\$}$ with 
    $x\neq\varepsilon$,
    thus falling in case $(C1)$.
   	\item $O_s=R_s, O_t=S_t$ with $t=s+2$  and $u[s+1]=v[s+1]$. This case is the symmetric of the previous one. Indeed, $f$ has a $2$-tilde-error overlap with non-adjacent error positions where the block $000 \choose 101$ may occur, but in this case $000$ is a suffix of $f$,    	
    	  thus falling again in case $(C1)$. 
    	\item $O_s=S_s, O_t=S_t$ with $t=s+3$  and $u[s+2]=v[s+2]$. 
    	Similarly as in the two previous cases, if 
    	$I_s\cap I_t$ includes $s+1, s+2, s+3$, but not $s$ and not $s+4$ then $f$ 
    	has a  $2$-tilde-error overlap with non-adjacent error positions where the block $000 \choose 101$ may occur, but in this case $000$ is a prefix of $f$ and $101$ is a 	suffix of $f$ thus falling again in case $(C1)$.    	
        \end{enumerate}
      \smallskip
    Suppose now that there is no character of $u$ interleaving the characters modified by $O_s$ and the ones modified by $O_t$.
    
    Again, four situations may occur.
    
    \begin{enumerate}
    \item $O_s=R_s, O_t=R_t$ with $t=s+1$. Then, $u[s]=u[s+1]$,
    otherwise $R_s$ and $R_t$ could be replaced in the tilde-transformation from $u$ to $v$ with
     a single swap $S_s$ obtaining a tilde-transformation with a less number of operations. Hence, suppose $u=w00z$; then,
     $O_s(u)=R_s(u)=w10z$ and $O_{s+1}(u)=R_{s+1}(u)=w01z$, and since the intervals $I_s$ and $I_t$ must share the positions $s$ and $t$, then $f$ has a $1$-tilde-error overlap ${\$x\choose ax}{01\choose 10}{ya\choose y\$}$ with $x,y\in \Sigma^*$, $a\in \Sigma$, as in (C0).

    \item \label{case2} $O_s=S_s, O_t=R_t$ with $t=s+2$. In this case $u=wc\overline{c}dz$, with $w,z\in \Sigma^*$ and $c,d\in\Sigma$. Suppose $c=1$ (then $\overline{c}=0$). By Lemma \ref{l-blocchi-esclusi}, $d\neq 1$, because
${u[s..s+2] \choose v[s..s+2]} \neq {101 \choose 010}$.
    On the other hand, if $d=0$ then 
    ${u[s..s+2] \choose v[s..s+2]} = {100 \choose 011}$ and $O_s(u)=w010z$ and $O_t(u)=w101z$.  
    If $I_s \cap I_t$ includes $s$, $s+1$ and $s+2$,  
    then (C3) is verified.
    If $I_s \cap I_t$ includes only $s+1$ and $t=s+2$, then $f$ has a $1$-tilde-error overlap 
    ${\$x\choose ax}{01\choose 10}{ya\choose y\$}$, 
    as in (C0).
   
    \item $O_s=R_s, O_t=S_t=S_{s+1}$. The proof is equivalent to the previous case by considering the reverses of $u$ and $v$ in the above proof.

    \item $O_s=S_s, O_t=S_t$ with $t=s+2$.
    In this case $u=wc\overline{c}d\overline{d}z$.
     Suppose $c=1$. Then two cases occur: $d=1$ and $d=0$. 
        \smallskip\\
        If $d=1$ then $u=w1010z$ and $O_s(u)=S_s(u)=w0110z$ and $O_t(u)=S_t(u)=S_{s+2}(u)=w1001z$. If all positions $s, s+1, t=s+2, t+1=s+3$ are included in $I_s \cap I_t$ then case (C2) is verified. 
        Consider now the case where $s, s+1, s+2 \in I_s \cap I_t $ and $s+3\not \in I_s \cap I_t$. This means that $I_s$ ends in position $s+2$, and $f$ has a $2$-tilde-error overlap 
        ${\$x\choose ax}{100\choose 011}{1\choose \$}$, as in case (C4), up to complement.
        \\In a similar way, if $s+1, s+2, s+3\in I_s \cap I_t$ and $s\not\in I_s \cap I_t$, then
        $I_t$ starts in position $s+1$ and 
         $f$ has a $2$-tilde-error overlap 
         ${\$\choose 0}{001\choose 110}{xa\choose x\$}$, as in case (C4), up to reverse and exchange of rows. 
        \\ The last situation is when  $s+1, s+2 \in I_s\cap I_t$ and $s, s+3\not\in I_s\cap I_t$ then $f$ starts with $001$ and ends with $011$, providing a $2$-tilde-error overlap as in case (C5).
    \smallskip\\
    Consider now $d=0$, then $u=w1001z$. Again, we have four cases. If $s, s+1, s+2, s+3\in I_s \cap I_t$ then case (C2) is verified.
    \\If $s, s+1, s+2\in I_s \cap I_t$ and $s+3\not\in I_s \cap I_t$, then $f$ has a $2$-tilde-error overlap
     ${\$x\choose ax}{101\choose 010}{0\choose \$}$, as in case (C3), with $y=\varepsilon$ and $b=0$ and up to complement.
    \\ If $s+1, s+2, s+3 \in I_s \cap I_t$ and $s \not \in I_s \cap I_t$, then $f$ has a 2-tilde-error overlap
     ${\$\choose 0}{010\choose 101}{yb\choose y\$}$, 
      as in case (C3), with $x=\varepsilon$ and $a=0$.
    \\ If $s+1, s+2 \in I_s \cap I_t$ and $s, s+3 \not \in I_s \cap I_t$, then $f$ starts with $01$ and ends with $10$ corresponding to case (C0) with $x, y=\varepsilon$. 
    
    \end{enumerate}\end{proof}

\subsection{Construction of tilde-witnesses}\label{s-witness}
As already discussed, in order to prove that a word is tilde-non-isometric it is sufficient to exhibit a pair of tilde-witnesses.
As a preparation for the sufficient condition of the characterization theorem we introduce a technique to define a pair  of tilde-witnesses for a word, starting from its tilde-error overlaps.

\begin{definition}
Let $f\in\Sigma^*$ have a $2$-tilde-error overlap of shift $r$, length $\ell=n-r$, and let $(O_i,O_j)$, $1\leq i<j \leq\ell$, be a tilde-transformation of ${\rm pre}_{\ell}(f)$ in ${\rm suf}_{\ell}(f)$,  with $(O_i,O_j) \neq (S_i, S_{i+1})$.

If $(O_i,O_j) \neq (R_i, S_{i+1}), (S_i, R_{i+2})$ then
\begin{equation} \label{eq:alfa_beta_tilde}
     \tilde\alpha_r(f)= {\rm pre}_r(f) O_i(f) \,\, \mbox{and} \,\, \tilde\beta_r(f)= {\rm pre}_r(f) O_j(f)
  \end{equation}
 \hspace{0.7cm} else
\begin{equation} \label{eq:alfa_beta_101}
     \tilde\alpha_r(f)= {\rm pre}_r(f) S_i(f) \,\, \mbox{and} \,\, \tilde\beta_r(f)= {\rm pre}_r(f) R_{i+2}(f)
  \end{equation}
\end{definition}

\noindent When there is no ambiguity, we omit $f$ and we simply write $\tilde\alpha_r$ and $\tilde\beta_r$.

\begin{lemma}\label{l-alfa-tilde-f-free}
    If $f$ has a $2$-tilde-error overlap of shift $r$  then $\tilde\alpha_r(f)$ is $f$-free.
\end{lemma}

\begin{proof}
Let $f$ have a $2$-tilde-error overlap of shift $r$ and length $\ell=n-r$ and let $(O_i,O_j)$, $1\leq i<j \leq\ell$, be a tilde-transformation of ${\rm pre}_{\ell}(f)$ in ${\rm suf}_{\ell}(f)$, with $(O_i,O_j) \neq (S_i, S_{i+1})$. If $(O_i,O_j)=(R_i, R_j)$ then $\tilde\alpha_r$ is $f$-free by Claim 1 of Lemma 2.2 in \cite{Wei17}, also in the case of adjacent errors. If $(O_i,O_j)=(S_i, R_{i+2})$ then, by Equation (\ref{eq:alfa_beta_tilde}), we have $\tilde\alpha_r=w_1S_i(f)$ then $\tilde\alpha_r[r+k]=f[k]$, for any $1 \leq k \leq n$, with $k \neq i$ and $k \neq i+1$. If $f$ occurs in $\tilde\alpha_r$ in position $r_1+1$ we have that $1 < r_1 < r$ (if $r_1=1$ then $f[i]=f[i+1]$ and this is not possible since $O_i=S_i$) and $\tilde\alpha_r[r_1+1 \ldots r_1+n]=f[1 \ldots n]$. Finally, we have that $\tilde\alpha_r[k]=f[k]$, for $k \neq r+j$. In conclusion, we have that $f[i]=\tilde\alpha_r[r_1+i]=f[r_1+i]$ (trivially, $r_1+i \neq r+j$). Furthermore $f[r_1+i]=\tilde\alpha_r[r+r_1+i]$ ($r_1+i \neq i$ and $r_1+i \neq i+1$ because $r_1 >1$). But $\tilde\alpha_r[r+r_1+i]=f[r+i] $ then we have the contradiction that $f[i]=f[r+i]$.

In all the other cases the proof is similar. For clarity, note that, also in the case of adjacent errors, supposing that $f$ occurs in $\tilde\alpha_r$ leads to a contradiction in $f[i]$ that is not influenced by $j$.\end{proof}

Note that while $\tilde\alpha_r$ is always $f$-free, $\tilde\beta_r$ is not. Indeed, the property of $\tilde\beta_r$ being not $f$-free is related to a condition on the overlap of $f$. We  
 give the following definition. 

\begin{definition}\label{d-alfabeta}
Let $f\in\Sigma^*$ have a $2$-tilde-error overlap of shift $r$, length $\ell=n-r$.
    Then, the $2$-tilde-error overlap satisfies $Condition^{\sim}$ 
if  any tilde-transformation $(O_i,O_j)$, $1\leq i<j \leq\ell$, of ${\rm pre}_{\ell}(f)$ in ${\rm suf}_{\ell}(f)$,
    is $(O_i,O_j) = (R_i, R_{j})$ or
    $(O_i,O_j) = (S_i, S_{j})$ and    
    \begin{center}
	\hspace{1 cm}
	$\left\{
	\begin{array}{ll}
 
			r \ \ is\ even &\\
			j-i = r/2 & \\
				f[i .. (i+r/2-1)] =f[j .. (j+r/2 -1)] & \\
		\end{array}
		\right. $
    \hspace{1 cm}	(\textbf{\textit{ Condition$^{\sim}$} })
    \end{center}
\end{definition}

    
\begin{lemma}\label{l-beta-tilde-f-free}
	Let $f\in \Sigma^n$ have a $2$-tilde-error overlap of shift $r$, then $\tilde\beta_r(f)$ is not $f$-free iff the  $2$-tilde-error overlap
	satisfies $Condition^{\sim}$.
	\end{lemma}

\begin{proof} 
Let us suppose $f$ have a $2$-tilde-error overlap of shift $r$ and length $\ell=n-r$ and let $(O_i,O_j)$, $1\leq i<j \leq\ell$, be a tilde-transformation of ${\rm pre}_{\ell}(f)$ in ${\rm suf}_{\ell}(f)$
that satisfies $Condition^{\sim}$.

Now, if $(O_i, O_j)=(R_i, R_j)$
  then the fact that $\tilde\beta_r(f)$ is not $f$-free can be shown as in the proof of Claim 2 of Lemma 2.2 in \cite{Wei17}. 
  
  If 
$(O_i, O_j)=(S_i, S_j)$, then that proof must be suitably modified as follows. 
Let $f[i]=f[j]=x$, $f[i+r]=f[j+r]=\overline{x}$, $f[i+1]=f[j+1]=\overline{x}$ and $f[i+1+r]=f[j+1+r]=x$. 
Then, it is possible to show that, for some $k_1, k_2 \geq 0$, 
$f= \rho (uw)^{k_1}uwuw\overline{u}w\overline{u}(w\overline{u})^{k_2}\sigma$, where $u=x\overline{x}$, $w=f[i+2 .. j-1]$ ($w$ is empty, if $j=i+2$) and
$\rho$ and $\sigma$ are, respectively, a suffix and a prefix of $w$.
Then, $\tilde\beta_r(f)=\rho (uw)^{k_1+1}uwuw\overline{u}w\overline{u}(w\overline{u})^{k_2+1}\sigma$ and, hence, $\tilde\beta_r(f)$ is not $f$-free. 
 
 Assume now that $\tilde\beta_r(f)$ is not $f$-free and suppose that a copy of $f$, say $f'$, occurs in $\tilde\beta_r(f)$ at position $r_1+1$. A reasoning similar to the one used in the proof of Lemma  \ref{l-alfa-tilde-f-free}, shows that 
 $j-i = r_1$ and $j-i = r-r_1$. Hence $r=2 r_1$ is even and $j-i=r/2$. Therefore, $f[i+t]=f'[i+t]= f[i+t+r/2]= f[j+t]$, for  $0 \leq t \leq r/2$, i.e. $f[i .. (i+r/2-1)] =f[j .. (j+r/2 -1)]$.
Moreover, such conditions with the assumption that $\tilde\beta_r(f)$ is not $f$-free can be used to show that
$(O_i, O_j)\neq(S_i, R_j)$ and $(O_i, O_j)\neq(R_i, S_j)$
and thus
the $2$-tilde-error overlap satisfies $Condition^{\sim}$.\end{proof}

  \subsection{The sufficient condition}\label{s:dritto}
In this last subsection we conclude the proof of Theorem \ref{t:characterization} by proving the   {\em if}  direction.

 We distinguish the cases of $1$-tilde-error overlap (case $(C0)$), $2$-tilde-error overlap with non-adjacent errors (case $(C1)$)  and,  finally, $2$-tilde-error overlap with adjacent errors (cases $(C2)-(C5)$). Non-adjacent errors can be dealt with the standard techniques used for the Hamming distance, while the case of adjacent ones may show new issues. Note  that, for some error types, we need also to distinguish sub-cases related to the different characters adjacent to the overlap.

For each case in the list, the proof consists in providing a pair of words which is proved to be a pair of tilde-witnesses for $f$ by showing that it fulfills conditions 1., 2. and 3. in Definition \ref{d-witnesses}.

Let us start with the case (C0) of a $1$-tilde-error overlap.

\begin{proposition}(C0)\label{p-1-sufficiente}
    If $f$ has a $1$-tilde-error overlap ${\$x \choose ax}{01 \choose 10}{yb \choose y\$}$, with $x,y \in \Sigma^*$, $a,b \in \Sigma$, then $f$ is tilde-non-isometric.
\end{proposition}
\begin{proof}
Let $f$ have a $1$-tilde-error overlap ${\$x \choose ax}{01 \choose 10}{yb \choose y\$}$ with shift $r$ and block ${01 \choose 10}$ occurring at position  $i=|x|+1$.
The pair $(u,v)$ with:
    $$u={\rm pre}_r(f)R_i(f) \hspace{1cm}v={\rm pre}_r(f)R_{i+1}(f)$$
is a pair of tilde-witnesses for $f$. 
Indeed, first of all, note that $(u,v)=(w11z, w00z)$ for some $w,z \in \Sigma^*$. Therefore,  $dist_\sim(u, v)=2$. Moreover, using techniques similar to the ones used in Lemma \ref{l-alfa-tilde-f-free}, one can prove that $u$ and $v$ are $f$-free.
Finally, no $f$-free tilde-transformation from $u$ to $v$ exists. The only possible tilde-transformations from $u$ to $v$ are $(R_{r+i}, R_{r+i+1})$ and $(R_{r+i+1},R_{r+i})$ and they both are not $f$-free since
$f$ occurs at position $r+1$ in
$R_{r+i}(u)$
and at position $1$ in $R_{r+i+1}(u)$.\end{proof}

\begin{example}
    The word $f=101$ has a $1$-tilde-error overlap ${\$ \choose 1}{10\choose 01}{1\choose \$}$  with block  ${10 \choose 01}$ in position $1$ therefore it is tilde-non-isometric. In fact, the pair $(u,v)$ with $u=1001$ and $v=1111$ is a pair of tilde-witnesses.
\end{example}

Let us consider the case of non-adjacent errors and prove that $f$ satisfies (C1) then it is tilde-non-isometric.
\begin{proposition}(C1)\label{p-skukkiati-C1}
If $f$ has a $2$-tilde-error overlap 
    with non-adjacent errors, 
different from   
    ${\$x \choose ax}{000\choose 101}{yb \choose y\$}$ with $x, y \in \Sigma^+$, $a,b\in \Sigma$ then $f$ is tilde-non-isometric.
\end{proposition}

\begin{proof}
Let $f$ have a  $2$-tilde-error overlap  with shift $r$ and length $\ell=n-r$ 
with non-adjacent errors
different from
  ${\$x \choose ax}{000\choose 101}{yb \choose y\$}$.
Let $(O_i,O_j)$, $1\leq i<j \leq\ell$ be a tilde-transformation from ${\rm pre}_{\ell}(f)$ to ${\rm suf}_{\ell}(f)$. 
In case of the $2$-tilde-error overlap ${\$x \choose ax}{100 \choose 001}{yb \choose y\$}$, consider $(O_i,O_j)=(R_i, R_{i+2})$.
Then,
if the $2$-tilde-error overlap 
does not
	satisfy $Condition^{\sim}$ then the pair $(\tilde\alpha_r, \tilde\beta_r)$ as in Equation (\ref{eq:alfa_beta_tilde}) is a pair of tilde-witnesses for $f$.
  Indeed, $dist_\sim(\tilde\alpha_r, \tilde\beta_r)=2$ and 
from Lemma \ref{l-alfa-tilde-f-free} and Lemma \ref{l-beta-tilde-f-free}, $\tilde\alpha_r$ and $\tilde\beta_r$ are $f$-free. 
Finally, no $f$-free tilde-transformation from $\tilde\alpha_r$ to $\tilde\beta_r$ exists. In fact,
 the tilde-transformations from $\tilde\alpha_r$ to $\tilde\beta_r$, $(O_{r+i},O_{r+j})$ and $(O_{r+j},O_{r+i})$, are not $f$-free because
$f$ occurs at position $r+1$ in
$O_{r+i}(\tilde\alpha_r)$
and at position $1$ in $O_{r+j}(\tilde\alpha_r)$. And no other 
tilde-transformation is possible since ${\tilde\alpha_r \choose \tilde\beta_r}$ cannot contain either $100\choose 001$ or $101 \choose 010$ (see Remark \ref{r-unici-blocchi-ambigui-dist2}). In fact  the $2$-tilde-error overlap of $f$ yealding $\tilde\alpha_r$ and $\tilde\beta_r$ would be  ${\$x \choose ax}{000\choose 101}{yb \choose y\$}$ in the first case, against the hypotesis,  and ${\$x \choose ax}{101\choose 010}{yb \choose y\$}$, in the second case, that is not allowed since the errors would be adjacent. 
\\
If, instead, the  $2$-tilde-error overlap
	satisfies $Condition^{\sim}$, then consider the pair $(\tilde\eta_r, \tilde\gamma_r)$ with $\tilde\eta_r={\rm pre}_r(f) O_i(f){\rm suf}_{r/2}(f)$, $\tilde\gamma_r={\rm pre}_r(f) O_j(O_t(f)){\rm suf}_{r/2}(f)$,where $t=j+r/2$ and $O_t=R_t$, if  $O_i=R_i$, and $O_t=S_t$, if  $O_i=S_i$.
 Note that $O_j(O_t(f))=O_t(O_j(f))$, since the $2$-tilde-error overlap has non-adjacent errors. The pair $(\tilde\eta_r, \tilde\gamma_r)$ is a pair of tilde-witnesses for $f$. Indeed, $dist_\sim(\tilde\gamma_r, \tilde\eta_r)=3$ (recall that the $2$-tilde-error overlap has non-adjacent errors). Moreover, using techniques similar to the ones used in Lemma \ref{l-alfa-tilde-f-free}, one can prove that $\tilde\gamma_r$ and $\tilde\eta_r$ are $f$-free. 
Finally, no $f$-free tilde-transformation from $\tilde\gamma_r$ to $\tilde\eta_r$ exists. In fact,
 any possible tilde-transformation from $\tilde\gamma_r$ to $\tilde\eta_r$, is given by a permutation of $\{O_{r+i},O_{r+j},O_{r+t}\}$ and any of such tilde-transformations let $f$ occur at position $1$, $r$ or $3r/2$.\end{proof}

Let us examine all the remaining cases (C2) up to (C5) that correspond to the different kinds of $2$-tilde-error overlaps of $f$ with adjacent errors and state the whole sufficient condition for a word to be non-tilde-isometric.

\begin{proposition}$[C2-C5]$\label{p-2-sufficiente}
Let $f \in \Sigma^n$. If  one of the following cases occurs for $f$ (up to complement, reverse, and exchange of rows) then $f$ is tilde-non-isometric.
\begin{itemize}
\item[$(C2)$]\label{case2Prop} $f$ has a $2$-tilde-error overlap ${\$x\choose ax}{0101\choose 1010}{yb\choose y\$}$ 
    or ${\$x\choose ax}{0110\choose 1001}{yb\choose y\$}$
    with $x,y\in \Sigma^*$, $a,b\in \Sigma$
    
    \item[$(C3)$]\label{case3Prop} $f$ has a $2$-tilde-error overlap 
    ${\$x\choose ax} {{010\choose 101}} {yb\choose y\$}$ with $x,y\in \Sigma^*$, $a,b\in \Sigma$
     
    \item[$(C4)$]\label{case4Prop} $f$ has a $2$-tilde-error overlap 
    ${\$x\choose ax}{011\choose 100}{0\choose \$}$
    with $x\in \Sigma^*$, $a\in \Sigma$
    
    \item[$(C5)$]\label{case5Prop} $f$ has a $2$-tilde-error overlap ${\$ \choose 0}{00\choose 11}{1\choose \$}$
\end{itemize}
\end{proposition}

\begin{proof}

\smallskip
\noindent

  \noindent (C2).

Let $f$ have a $2$-tilde-error overlap ${\$x\choose ax}{0101\choose 1010}{yb\choose y\$}$ 
    or ${\$x\choose ax}{0110\choose 1001}{yb\choose y\$}$
    with $x,y\in \Sigma^*$, $a,b\in \Sigma$. 

We distinguish the proof in two cases.\\
\noindent 
{\em Case a.} Suppose that $f$ has a $2$-tilde-error overlap ${\$x\choose ax}{0101\choose 1010}{yb\choose y\$}$, $x,y\in \Sigma^*$, $a,b\in \Sigma$,  with shift $r$ and length $\ell=n-r$ and let $(O_i,O_j)$, be a tilde-transformation from ${\rm pre}_{\ell}(f)$ to ${\rm suf}_{\ell}(f)$. Note that it must be $(O_i,O_j)=(S_i,S_{i+2})$ or $(O_i,O_j)=(S_{i+2},S_i)$ where $i=|x|+1$.

If the $2$-tilde-error overlap does not satisfy $Condition^{\sim}$ then the pair $(\tilde\alpha_r, \tilde\beta_r)$ as in Equation (\ref{eq:alfa_beta_tilde}) is a pair of tilde-witnesses for $f$.
  Indeed, $dist_\sim(\tilde\alpha_r, \tilde\beta_r)=2$ and 
from Lemma \ref{l-alfa-tilde-f-free} and Lemma \ref{l-beta-tilde-f-free}, $\tilde\alpha_r$ and $\tilde\beta_r$ are $f$-free. 
Finally, no $f$-free tilde-transformation from $\tilde\alpha_r$ to $\tilde\beta_r$ exists. In fact,
 the only possible tilde-transformations from $\tilde\alpha_r$ to $\tilde\beta_r$ are $(S_{r+i},S_{r+i+2})$ and $(S_{r+i+2},S_{r+i})$. But they are not $f$-free because
 they let $f$ occur at position $1$ and at position $r+1$ in 
 $S_{r+i+2}(\tilde\alpha_r)$ 
 and $S_{r+i}(\tilde\alpha_r)$, respectively.
\\  
If instead the  $2$-tilde-error overlap
	satisfies $Condition^{\sim}$, then consider the pair $(\tilde\eta_r, \tilde\gamma_r)$ with $\tilde\eta_r={\rm pre}_r(f) S_i(f){\rm suf}_{r/2}(f)$, $\tilde\gamma_r={\rm pre}_r(f) S_{i+2}(S_{i+4}(f)){\rm suf}_{r/2}(f)$.
Note that $S_{i+2}(S_{i+4}(f)=S_{i+4}(S_{i+2}(f)$.

 The pair $(\tilde\eta_r, \tilde\gamma_r)$ is a pair of tilde-witnesses for $f$. Indeed, it is easy to see that $dist_\sim(\tilde\gamma_r, \tilde\eta_r)=3$. Moreover, using techniques similar to the ones used in Lemma \ref{l-alfa-tilde-f-free}, one can prove that $\tilde\gamma_r$ and $\tilde\eta_r$ are $f$-free. 
Finally, no $f$-free tilde-transformation from $\tilde\gamma_r$ to $\tilde\eta_r$ exists. In fact,
 any possible tilde-transformations  from $\tilde\gamma_r$ to $\tilde\eta_r$, consists of three swap operation and is given by a permutation of $\{S_{r+i},S_{r+i+2},S_{r+i+4}\}$. But any of such tilde-transformations let $f$ occurs at position $1$ or $r+1$ or $3r/2+1$.

\noindent
{\em Case b.}  Suppose that $f$ has a $2$-tilde-error overlap ${\$x\choose ax}{0110\choose 1001}{yb\choose y\$}$, $x,y\in \Sigma^*$, $a,b\in \Sigma$,  with shift $r$ and length $\ell=n-r$ and let $(O_i,O_j)$, be a tilde-transformation from ${\rm pre}_{\ell}(f)$ to ${\rm suf}_{\ell}(f)$. 
Note that it must be $(O_i,O_j)=(S_i,S_{i+2})$ or $(O_i,O_j)=(S_{i+2},S_i)$ where $i=|x|+1$.
Now, the pair $(\tilde\alpha_r, \tilde\beta_r)$ as in Equation (\ref{eq:alfa_beta_tilde}) is a pair of tilde-witnesses for $f$.
Note that, in this case, the $2$-tilde-error overlap does not satisfy Condition$^\sim$ since $j=i+2$ and $f[i]$ is different from $f[j]$.

	\smallskip
\noindent (C3)
Let $f$ have a $2$-tilde-error overlap 
    ${\$x\choose ax}{{010\choose 101}} {yb\choose y\$}$, $x,y\in \Sigma^*$, $a,b\in \Sigma$ with shift $r$ and length $\ell=n-r$ and let $(O_i,O_j)$, be a tilde-transformation from ${\rm pre}_{\ell}(f)$ to ${\rm suf}_{\ell}(f)$. Note that it must be $(O_i,O_j)=(S_i,R_{i+2})$ or $(O_i,O_j)=(R_i,S_{i+1})$. 
Then,
the $2$-tilde-error overlap 
does not
	satisfy $Condition^{\sim}$. The pair $(\tilde\alpha_r, \tilde\beta_r)$ as in Equation (\ref{eq:alfa_beta_101}) is a pair of tilde-witnesses for $f$.
  Indeed, it is easy to see that $dist_\sim(\tilde\alpha_r, \tilde\beta_r)=2$. Moreover, 
from Lemma \ref{l-alfa-tilde-f-free} and Lemma \ref{l-beta-tilde-f-free}, $\tilde\alpha_r$ and $\tilde\beta_r$ are $f$-free. 
Finally, no $f$-free tilde-transformation from $\tilde\alpha_r$ to $\tilde\beta_r$ exists. In fact,
 the tilde-transformations from $\tilde\alpha_r$ to $\tilde\beta_r$, $(O_{r+i},O_{r+j})$ and $(O_{r+j},O_{r+i})$, are not $f$-free because
 they let $f$ occur at position $r+1$ and at position $1$ in $O_{r+i}(\tilde\alpha_r)$ and $O_{r+j}(\tilde\alpha_r)$, respectively, and no other 
tilde-transformation is possible.

\smallskip
\noindent (C4) 
Let $f$ have a $2$-tilde-error overlap
${\$x\choose ax}{011\choose 100}{0\choose \$}$, $x\in \Sigma^*$, $a\in \Sigma$ with shift $r$ and length $\ell=n-r$ and let $(O_i,O_j)$, be a tilde-transformation from ${\rm pre}_{\ell}(f)$ to ${\rm suf}_{\ell}(f)$. 
Note that it must be $(O_i,O_j)=(S_i,R_{i+2})$ with $i=|x|+1$. 
 In this case we need a different technique to construct the pair of tilde-witnesses $(\tilde\alpha_r,\B_r)$ for $f$. 
More exactly, we set $\tilde\alpha_r$ as in Equation (\ref{eq:alfa_beta_101}) and $\B_r={\rm pre}_r(f) S_{i+2}(f)$. 
It is easy to see that $dist_\sim(\tilde\alpha_r, \B_r)=2$. Moreover, from Lemma \ref{l-alfa-tilde-f-free}, $\tilde\alpha_r$ is $f$-free. Here we prove that $\B_r$ is f-free. 
Indeed, suppose that $f$ occurs in  $\B_r$ starting from position $r_1+1$. 
By the definition of $\B_r$, we have that $\B_r[k]=f[k]$ for $1 \leq k \leq |f|$ and $k \neq r+i,r+i+1$. Moreover, $\B_r[r+k]=f[k]$, for $1 \leq k \leq |f|$ and $k \neq j,j+1$. Now, we have $0=f[i]$ and, since $f$ occurs in $\B_r$ starting from position $r_1+1$, $f[i]=\B_r[r_1+i]$.  Since $\B_r[r+i+1]=1$, it holds $r_1+i \neq r+i+1$. Moreover $r_1+i \neq r+i$ and we can conclude that $\B_r[r_1+i]=f[r_1+i]$ i.e. $0=f[i]=\B_r[r_1+i]=f[r_1+i]$. 
Note that it cannot be $r_1+i=j+1$, since $\B_r[r+j+1]=1$. Now, if $r_1+i \neq j$, we can say $f[r_1+i]=\B_r[r+r_1+i]$. Since $f$ occurs in  $\B_r$ starting from position $r_1+1$, we have $\B_r[r+r_1+i]=f[r+i]$
Altogether, we have $0=f[i]=\B_r[r_1+i]=f[r_1+i]=\B_r[r+r_1+i]=f[r+i]=1$, a contradiction. Therefore, the only possibility is that $r_1+i = j$, i.e. $r_1=2$.
Similar considerations on $f[j]$ show that, if $f$ occurs in  $\B_r$ starting from position $r_1+1$, then $r_1=1$. 
Therefore, if a copy of $f$ occurs in  $\B_r$ starting from position $r_1+1$, then it should be $r_1=2$ and $r_1=1$ and this is impossible. Hence $\B_r$ is $f$-free. 
\\
Moreover, no $f$-free tilde-transformation from $\tilde\alpha_r$ to $\B_r$ exists. The only possible tilde-transformations are $(S_{r+i}, S_{r+i+2})$ and $(S_{r+i+2},S_{r+i})$ and they are not $f$-free.
Hence $(\tilde\alpha_r,\B_r)$ is a pair of tilde-witnesses for $f$.

\smallskip
 \noindent (C5)
 Let $f$ have a $2$-tilde-error overlap ${\$ \choose 0}{00\choose 11}{1\choose \$}$ with shift $r$ and length $2$ and let $(O_i,O_j)$, be a tilde-transformation from ${\rm pre}_{2}(f)=00$ to ${\rm suf}_{2}(f)=11$. 
Note that it must be $(O_i,O_j)=(R_i,R_{i+1})$ with $i=1$. 
 In this case we have
$f=w001=001z$, for some $w,z \in \Sigma^*$. 
Consider the pair  $(\tilde\alpha_r, \tilde\psi_r)$ with $\tilde\alpha_r$ as in Equation (\ref{eq:alfa_beta_tilde}) and $\tilde\psi_r={\rm pre}_{r-1}(f)1 S_{i+1}(f)$. In other words 
$\tilde\alpha_r=w0101z$ and
$\tilde\psi_r=w1010z$.

The pair $(\tilde\alpha_r, \tilde\psi_r)$ is a pair of tilde-witnesses for $f$. Indeed, it is easy to see that $dist_\sim(\tilde\alpha_r, \tilde\psi_r)=2$. Moreover, from Lemma \ref{l-alfa-tilde-f-free}, $\tilde\alpha_r$ is $f$-free and, using techniques similar to the ones used in (C4) for $\B_r$, one can prove that $\tilde\psi_r$ is $f$-free. 
Finally, no $f$-free tilde-transformation from $\tilde\alpha_r$ to $\tilde\psi_r$ exists. The only possible tilde-transformations are $(S_{r-1}, S_{r+1})$ and $(S_{r+1},S_{r-1})$ and they both are not $f$-free.
Remark that, in this case, the pair $(\tilde\alpha_r,\tilde\beta_r)$ of Equation \ref{eq:alfa_beta_tilde} is not a pair of tilde-witnesses because ${\rm dist}_\sim(\tilde\alpha_r,\tilde\beta_r)=1$. \end{proof}

The following is an example of construction of tilde-witnesses when $f$ has a $2$-tilde-error overlap with two adjacent swaps.

\begin{example}\label{e-2swap} 
The word $f=10010110$ has a $2$-tilde-error overlap ${\$\choose 1}{1001 \choose 0110}{0\choose \$}$ of shift $r=4$ and swaps in positions $1,3$. By Proposition 
\ref{p-2-sufficiente} case (C2), the pair $(\tilde{\alpha}_4, \tilde{\beta}_4)$ with $\tilde{\alpha}_4=100101010110$ and $\tilde{\beta}_4=100110100110$ is a pair of tilde-witnesses. Then $f$ is tilde-non-isometric. 
Note that $f$ is Ham-isometric.
\end{example}

Propositions \ref{p-1-sufficiente}, \ref{p-skukkiati-C1} and \ref{p-2-sufficiente} together, prove the sufficient condition for a word to be tilde-non-isometric.
This concludes the proof of Theorem \ref{t:characterization}.

\section{Conclusion}

In this paper the problem of fully characterizing isometric words has been solved in the setting of the edit distance that allows swap and replacement operations. Compared with the setting of Hamming distance, the problem turned out to be much more complicated, since swap operation gives rise to transformations that differ from one another not only because of the execution order of operations, as in the case of Hamming, but also for the kind of operations executed. This fact makes things harder to manage and the characterization based on overlaps needs to specify some special cases that can occur.

Recall that the original reason for the introduction of isometric words was their connection with isometric subgraphs of the hypercube $Q_n$ (cf.\cite{IlicK12,Hsu93,KlavzarS12}).
In fact, if all the nodes containing a given isometric word as factor are removed then the distances between the remaining nodes do not change. These subgraph of hypercubes are called generalized Fibonacci cubes and are subgraphs isometric to hypercubes in the graph-theoretic sense. 

Over the years, many variations of the hypercube have been introduced in order to improve some of its features. For example, folded hypercubes (cf. \cite{Ama91}) and ehnanced hypercubes (cf. \cite{TzengW91}) have been defined by adding some edges to the hypercube and present many advantages for some topological features.
In \cite{ACFGMMDcfs23}, the definition of tilde-hypercube is introduced by exploiting the tilde-distance. Adding the swap operation increases the number of connections in the hypercube (since some nodes having distance $2$ in the hypercube, using swap operation,  have distance $1$ in the tilde-hypercube). Tilde-isometric words then are used to introduce the generalized tilde-Fibonacci cubes that are isometric subgraphs of the tilde-hypercube and revealed important properties. 
 
 \smallskip
In conclusion, the swap and mismatch distance we adopted in this paper opens up new scenarios and presents interesting new situations that surely deserve further investigation. The definition and characterization of tilde-isometric words, besides having an important combinatorial value can serve as base for strings and graph algorithmic developments.

\bibliographystyle{ws-ijfcs}
\bibliography{references}

 \end{document}